\documentclass[12pt,a4paper,oneside,reqno]{amsart}

\usepackage{amssymb,amsmath}
\usepackage{graphicx,psfrag}
\usepackage{stmaryrd} 

\newcommand{\parrow}{\xrightarrow{\resizebox{!}{3.5pt}{$\circ$}}}
\newcommand{\Ob}{\mathrm{Ob}}
\newcommand{\Ph}{\mathrm{Ph}}
\newcommand{\llangle}{\mbox{$\langle\hspace*{-3pt}\langle$}}
\newcommand{\rrangle}{\mbox{$\rangle\hspace*{-3pt}\rangle$}}

\newcommand{\vr}{\vec{r}}
\newcommand{\vp}{\vec{p}}
\newcommand{\vpp}{\vec{p}\,}
\newcommand{\vq}{\vec{q}}
\newcommand{\vqq}{\vec{q}\,}

\newcommand{\vy}{\vec{y}}

\newcommand{\vc}{\vec{c}}

\newcommand{\vo}{\vec{o}}

\newcommand{\vvh}{\vec{v}^{\,h}}
\newcommand{\vah}{\vec{a}^{\,h}}
\newcommand{\vvkm}{\vec{v}^{\,k}_m}

\newcommand{\vakm}{\vec{a}^{\,k}_m}
\newcommand{\vet}{\vec{1}_t}
\newcommand{\vex}{\vec{1}_x}

\newcommand{\Q}{\mathrm{Q}}
\newcommand{\B}{\mathrm{B}}
\newcommand{\W}{\mathrm{W}}

\newcommand{\IOb}{\mathrm{IOb}}

\renewcommand{\time}{\mathsf{time}}
\newcommand{\dist}{\mathsf{dist}}
\newcommand{\leteq}{\,\mbox{$:=$}\,}
\newcommand{\mort}{\bot_\mu}
\newcommand{\upp}{\uparrow\hspace*{-1pt}\uparrow\!}
\newcommand{\com}{\succ}
\newcommand{\pheq}{\,\lambda\,}
\newcommand{\seq}{\,\sigma\,}
\newcommand{\teq}{\,\tau\,}
\newcommand{\simrad}{\thicksim^{rad}}
\newcommand{\simph}{\thicksim^{ph}}
\newcommand{\simmu}{\thicksim^\mu}
\newcommand{\then}{\enskip \Longrightarrow\ }
\newcommand{\rship}{\mbox{$>\hspace{-6pt}\big|$}b,k,c\mbox{$\big>_{\!rad}$}}
\newcommand{\mship}{\mbox{$>\hspace{-6pt}\big|$}b,k,c \mbox{$\big>_{\!\!\mu}$}}
\newcommand{\ship}{\mbox{$>\hspace{-6pt}\big|$}b,k,c \big>}
\renewcommand{\and}{\;\land\;}
\newcommand{\setclose}{\}}
\newcommand{\Setclose}{\,\right\}}
\newcommand{\setopen}{\{}
\newcommand{\Setopen}{\left\{\,}

\newcommand{\phsum}{\rightthreetimes}
\newcommand{\dom}{Dom\,}
\renewcommand{\d}{\mathit{d}}
\newcommand{\ran}{Ran\,}
\newcommand{\?}{\textbf{}}

\usepackage{color} 
\definecolor{thmcolor}{rgb}{0,0,.4} 
\definecolor{remarkcolor}{rgb}{0,.2,0} 
\definecolor{proofcolor}{rgb}{.4,0,0} 
\definecolor{quecolor}{rgb}{.2,.2,0} 
\definecolor{axcolor}{rgb}{.23,0,.23}
\definecolor{axbgcolor}{rgb}{1,.6,1} 
\definecolor{defbgcolor}{rgb}{0.9,0.8,0.1} 
\definecolor{thmbgcolor}{rgb}{0.8,0.8,1} 
\definecolor{rmbgcolor}{rgb}{0.7,1,0.7} 
\definecolor{proofbgcolor}{rgb}{1,0.7,0.7} 
\definecolor{lightred}{rgb}{1,0.3,0.3}

\newcommand{\ax}[1]{\textcolor{axcolor}{\ensuremath{\mathsf{#1}}}} 
\newcommand{\Ax}[1]{\textcolor{axcolor}{\colorbox{axbgcolor}{\ensuremath{\mathsf{#1}}}}} 

\newcommand{\df}[1]{{\bf #1}} 

\newcommand{\Dtf}[1]{\setlength{\fboxsep}{2pt}\colorbox{defbgcolor}{#1}\setlength{\fboxsep}{3pt}} 
\newcommand{\Df}[1]{\setlength{\fboxsep}{2pt}\colorbox{defbgcolor}{\ensuremath{#1}}\setlength{\fboxsep}{3pt}} 
\newcommand{\Dff}[1]{\setlength{\fboxsep}{0pt}\colorbox{defbgcolor}{\ensuremath{#1}}\setlength{\fboxsep}{3pt}}

\theoremstyle{definition} \newtheorem{thm}{\colorbox{thmbgcolor}{\textcolor{thmcolor}{Theorem}}}[section] 
\theoremstyle{definition}  
\theoremstyle{definition} \newtheorem{lem}[thm]{\colorbox{thmbgcolor}{\textcolor{thmcolor}{Lemma}}}%
\theoremstyle{definition} \newtheorem{prop}[thm]{\colorbox{thmbgcolor}{\textcolor{thmcolor}{Proposition}}}
\theoremstyle{remark} \newtheorem{conv}[thm]{\colorbox{rmbgcolor}{\sc\textcolor{remarkcolor}{Convention}}} 
\theoremstyle{remark} 
\theoremstyle{definition}  
\theoremstyle{definition} \newtheorem{rem}[thm]{\colorbox{rmbgcolor}{\textcolor{remarkcolor}{Remark}}}

\begin{document}

\title[A logical analysis of the time-warp effect of GR]{A logical analysis of the time-warp effect of general relativity}

\author{Judit X.\ Madar\'asz, Istv\'an N\'emeti and Gergely Sz\'ekely}

\date{2007-08-14}

\begin{abstract}
Several versions of the Gravitational Time Dilation effect of General Relativity are formulated by the use of Einstein's Equivalence Principle.
It is shown that all of them are logical consequence of a first-order axiom system of Special Relativity extended to accelerated observers.
\end{abstract}

\maketitle

\section{Introduction}%

Our general aim is to turn spacetime theories into axiomatic theories of First-Order Logic (\Dtf{FOL}) and exhaustively investigate the relationship between the axioms and their consequences.

Why is it useful to apply the axiomatic method to Relativity Theory?
For one thing, this method makes it possible for us to understand the role of any particular axiom (that is, a basic assumption of the theory).
We can check what happens to the theory if we drop, weaken or replace the axiom by its negation.
For instance, it is shown by this method that the impossibility of faster than light motion is not independent from the other assumptions of Special Relativity (\Dtf{SR}), see \cite{AMNsamp}, \cite[\S 3.4]{pezsgo}.
(More boldly: it is superfluous as an axiom because it is provable as a theorem from much simpler and more convincing basic assumptions.)
The linearity of the transformations between observers (reference frames) can also be proved from some plausible assumptions, see \cite{AMNsamp}, \cite{pezsgo} and Theorem \ref{thm-poi}.
Moreover, we can discover new, interesting and physically relevant theories by this method.
This happened in the case of the axiom of parallels in Euclid's geometry; this kind of investigation led to the discovery of hyperbolic geometry.

Moreover, if we have an axiom system, we can ask which axioms are responsible for a certain consequence of the theory.
This kind of reverse thinking can help us to answer the why-type questions of Relativity.
For example, we can take the Twin Paradox and check which axiom of SR was and which one was not needed to derive it.
The weaker an axiom system is, the better answer it offers to the question: Why is the Twin Paradox true?.
For more details on this kind of investigation into the Twin Paradox, see \cite{Twp, mythes}.
We hope that we have given good reasons why we use the axiomatic method in our research into spacetime theories.
For more details or further reasons, see, e.g., Guts~\cite{guts}, Schutz~\cite{schutz}, Suppes~\cite{suppes}.

So far we have not said anything about why choosing FOL instead of the so powerful Second-Order Logic or any other abstract logic.
The main reason comes from the fact that we would like to use an \emph{absolute}\footnotemark\ logic for our investigations because obviously we do not want the consequence relation of the used logic to depend on Set Theory.
That is clear since our main subject is this relation; 
hence we want to understand its properties as clearly as possible, that is, as independently from Set Theory as possible.
We would also like to use a {\em complete}\footnotemark[\value{footnote}] logic since we would like to know that if something is true in all the possible models, it is also provable.
By Lindstr\"om's theorem, FOL is the \emph{strongest}\footnotemark[\value{footnote}] possible compact logic with L\"ovenheim-Skolem property, see, e.g., \cite{flum}.
Obviously {\em compactness}\footnotemark[\value{footnote}] follows from completeness.
V\"a\"an\"anen has proved that absolute logics have the \emph{L\"ovenheim-Skolem property}\footnotemark[\value{footnote}], see \cite{vaananen85}.
Thus we do not have any better candidate than FOL for our work.
\footnotetext{For precise definition of these concepts, see, e.g., \cite{MTL}.}
For further details of this reason or for other reasons for choosing FOL for axiomatic foundation, see, e.g., Ax~\cite{Ax}, \cite[\S ``Why FOL?'']{pezsgo}, V\"a\"an\"anen \cite{vaananen}, Wole\'nski \cite{wolenski}.

In this paper we concentrate on a well-known consequence of General Relativity (\Dtf{GR}), the Gravitational Time Dilation (\Dtf{GTD}).
GTD roughly says ``gravitation makes time flow slower.''
Here we investigate the relationship of GTD and a version of SR extended with
accelerated observers (thus extended for simulating gravity).
We use Einstein's Equivalence Principle (\Dtf{EEP}) to treat gravitation in SR.
EEP roughly says that ``a uniformly accelerated frame of reference is indistinguishable from a rest frame in a uniform gravitational field,'' see, e.g., Einstein \cite{Einstein11} or d'Inverno~\cite[\S 9.4]{d'Inverno}.
So instead of gravitation we talk about acceleration.
To investigate GTD in FOL, we have to fix a language (a set of basic concepts), present one or more axiom systems of SR and formulate GTD in this fixed language.
Then we can investigate the connection between GTD and the axiom systems by proving theorems and providing counterexamples.
As an illustration of our research, we have partly fulfilled this task in \cite{FOLfoundRT}.
In this paper after recalling the axiom systems, definitions and theorems presented in \cite{FOLfoundRT}, we concentrate on proving these theorems and developing the necessary tools to do so.
Although we develop the most important tools and prove most of the theorems that we have stated in \cite{FOLfoundRT}, we do not go into every detail, and do not prove all the theorems stated in \cite{FOLfoundRT} because that would make our paper too long.
We try to be as self-contained as possible.
First occurrences of concepts used in this work are set in boldface to make them easier to find.
We also use colored text and boxes to help the reader to find the axioms, notations, etc.
Throughout this work, if-and-only-if is abbreviated to {\bf iff}.

\section{A first-order axiom system of SR extended with accelerated observers}%
\label{ax-sec}

Let us now recall our first-order language and some 
of our axiom systems for SR.

The motivation for our basic concepts is summarized as follows.
Here we only deal with the kinematics of relativity, that is, we deal
 with motion of {\em bodies} (test-particles).
We represent motion as changing spatial location in time.
To do so, we have reference-frames for coordinatizing events (sets of bodies).
{\em Quantities} are used for marking time and space.
The structure of quantities is assumed to be an ordered field in place of the field of real numbers.
For simplicity, we associate reference-frames with certain bodies 
 called {\em observers}.
This observation is coded by the \emph{world-view relation}.
We visualize an observer as ``sitting'' in the origin of the
space part of its reference-frame, or equivalently, ``living'' on the 
time-axis of the reference-frame.
We distinguish {\em inertial observers} from the others.
We also use another special kind of bodies called {\em photons}.

Allowing ordered fields in place of the field of
reals increases the flexibility of our theory and minimizes the
amount of our mathematical presuppositions.
For further motivation in this direction, see, e.g., Ax~\cite{Ax}.
Similar remarks apply to our flexibility-oriented decisions below, for example, the one to treat the dimension of spacetime as a variable.

Using observers in place of coordinate systems or reference frames is only a matter of didactic convenience and visualization.
There are many reasons for using observers (or coordinate systems, or reference-frames) instead of a single
observer-independent spacetime structure.
One of them is that it helps us to weed unnecessary axioms from our
theories; but we state and emphasize the equivalence/duality between
observer-oriented and observer-independent approaches to relativity
theory, see \cite[\S 3.6]{logst}, \cite[\S 4.5]{Mphd}.
Motivated by the above, 
now we turn to fixing the first-order language of our axiom systems.

First we fix a natural number $\Df{d}\ge 2$ for the dimension of spacetime.
Our language contains the following non-logical symbols:
\begin{itemize}
\item unary relation symbols \Df{\B} (for \df{bodies}), \Df{\Ob} (for \df{observers}),
\Df{\IOb} (for \df{inertial observers}), \Df{\Ph} (for \df{photons}) and $\Df{\Q}$ (for \df{quantities}),
\item binary function symbols \Df{+}, \Df{\cdot} and a binary relation symbol \Df{\le} (for the field operations and the ordering on
$\Q$) and
\item a $2+d$-ary relation symbol \Df{\W} (for \df{world-view relation}).
\end{itemize}

We read $\B(x)$, $\Ob(x)$, $\IOb(x)$, $\Ph(x)$ and $\Q(x)$ as ``$x$ is a body,''
``$x$ is an observer,'' ``$x$ is an inertial observer,'' ``$x$ is a photon,'' ``$x$ is a
quantity.''
We use the world-view relation $\W$ to talk about coordinatization by reading
$\W(x,y,z_1,\ldots, z_d)$ as ``observer $x$ coordinatizes
body $y$ at spacetime location $\langle z_1,\ldots,z_d\rangle$,'' (that is, at space location $\langle z_2,\ldots,z_d\rangle$ at instant $z_1$).

$\B(x)$, $\Ob(x)$, $\IOb(x)$, $\Ph(x)$, $\Q(x)$, $\W(x,y,z_1,\ldots, z_d)$, $x=y$ and $x\leq y$
are the so-called {atomic formulas} of our first-order
language, where $x,y,z_1,\dots,z_d$ can be arbitrary variables or terms built
up
from variables by using the field operations.
 The \df{formulas} of our first-order language are built up from these
atomic formulas by using the logical connectives {\em not}
(\Df{\lnot}), {\em and} (\Df{\land}), {\em or} (\Df{\lor}), {\em implies}
(\Df{\Longrightarrow}), {\em if-and-only-if} (\Df{\Longleftrightarrow}), and the
quantifiers {\em exists} $x$ (\Df{\exists x}) and {\em for all $x$} (\Df{\forall x})
for every variable $x$.

The \df{models} of this language are of the form
\begin{equation}
\Df{\mathfrak{M}} = \langle U; \B, \Ob, \IOb, \Ph, \Q,+,\cdot,\leq,\W\rangle,
\end{equation}
where $U$ is a nonempty set, and $\B$, $\Ob$, $\IOb$, $\Ph$ and $\Q$ are unary relations on $U$, etc.
A unary relation on $U$ is just a subset of $U$.
Thus we use $\B$, $\Ob$, etc.\ as sets as well, for example, we write \Df{m\in \Ob} in place of $\Ob(m)$.

We use the notation $\Df{\Q^n}\leteq\Q\times\ldots\times \Q$ ($n$-times) for the set of all $n$-tuples of elements of $\Q$.
If $\vpp\in \Q^n$, we assume that $\Df{\vpp}=\langle p_1,\ldots,p_n\rangle$, that is, $p_i\in\Q$ denotes the $i$-th component of the $n$-tuple $\vpp$.
We write $\W(m,b,\vpp)$ in place of
$\W(m,b,p_1,\dots,p_d)$, and we write $\forall \vpp$ in place of
$\forall p_1,\dots,\forall p_d$, etc.

We present each axiom at two levels.
First we give an intuitive formulation, 
then a precise formalization using our logical notations 
(which can easily be translated into first-order formulas by inserting the first-order
definitions into the formalizations).
We aspire to formulate easily understandable axioms in FOL.

The first axiom expresses our very basic assumptions, such as: 
both photons and observers are bodies, 
inertial observers are also observers, etc.

\begin{description}
\item[\Ax{AxFrame}] $\Ob\cup \Ph\subseteq \B$, $\IOb\subseteq \Ob$, $\W\subseteq \Ob \times
\B\times \Q^d$, $\B\cap \Q=\emptyset$; $+$ and $\cdot$ are binary operations and $\le$
is a binary relation on $\Q$.
\end{description}
Instead of this axiom we could also use many-sorted first-order language as in \cite{pezsgo} and \cite{logst} and only assume that $\IOb\subseteq\Ob$.

To be able to add, multiply and compare measurements of observers,
we put algebraic structure on the set of quantities by the
next axiom.

\begin{description}
\item[\Ax{AxEOF}]
A first-order axiom saying the \df{quantity part} $\left< \Q;
+,\cdot, \le \right>$ is a Euclidean ordered field, that is, a linearly ordered field in which positive
elements have square roots.
\end{description}
For the first-order definition of linearly ordered field, see, e.g., \cite{Chang-Keisler}.
We use the usual first-order definable field operations \Dff{0, 1, -, /, \sqrt{\phantom{i}}}.
We also use the vector-space structure of $\Q^n$, 
that is, if $\vpp,\vqq\in \Q^n$ and $\lambda\in \Q$, then $\Df{\vpp+\vqq, -\vpp, \lambda\cdot\vpp}\in \Q^n$; 
and $\Df{\vo}\,\leteq\langle 0,\ldots,0\rangle$ denotes the \df{origin}.

\begin{conv}
We treat \ax{AxFrame} and \ax{AxEOF} as a part of our logical frame.
Hence without any further mentioning, they are always assumed and will be part of each axiom system we propose herein.
\end{conv}

\begin{figure}[h!btp]
\small
\begin{center}
\psfrag{mm}[bl][bl]{$wl_m(m)$}
\psfrag{mk}[br][br]{$wl_m(k)$}
\psfrag{mb}[t][t]{$wl_m(b)$}
\psfrag{mph}[tl][tl]{$wl_m(ph)$}
\psfrag{kk}[br][br]{$wl_k(k)$}
\psfrag{km}[b][bl]{$wl_k(m)$}
\psfrag{kb}[t][t]{$wl_k(b)$}
\psfrag{kph}[tl][tl]{$wl_k(ph)$}
\psfrag{p}[r][r]{$\vpp$}
\psfrag{k}[bl][bl]{$k$}
\psfrag{b}[tl][tl]{$b$}
\psfrag{m}[bl][bl]{$m$}
\psfrag{ph}[tl][tl]{$ph$}
\psfrag{evm}[tl][tl]{$ev_m$}
\psfrag{evk}[bl][bl]{$ev_k$}
\psfrag{Evm}[r][r]{$Ev_m$}
\psfrag{Evk}[r][r]{$Ev_k$}
\psfrag{Ev}[r][r]{$Ev$}
\psfrag{T}[l][l]{$e=ev_m(\vpp)=ev_k(\vqq)$}
\psfrag{q}[r][r]{$\vqq$}
\psfrag{Cdk}[l][l]{$Cd_k$}
\psfrag{Cdm}[l][l]{$Cd_m$}
\psfrag{Crdk}[l][l]{$Crd_k$}
\psfrag{Crdm}[l][l]{$Crd_m$}
\psfrag{t}[lb][lb]{$$}
\psfrag{o}[t][t]{$\vo$}
\psfrag{fkm}[t][t]{$w^k_m$}
\psfrag*{text1}[cb][cb]{world-view of $k$}
\psfrag*{text2}[cb][cb]{world-view of $m$}
\includegraphics[keepaspectratio, width=\textwidth]{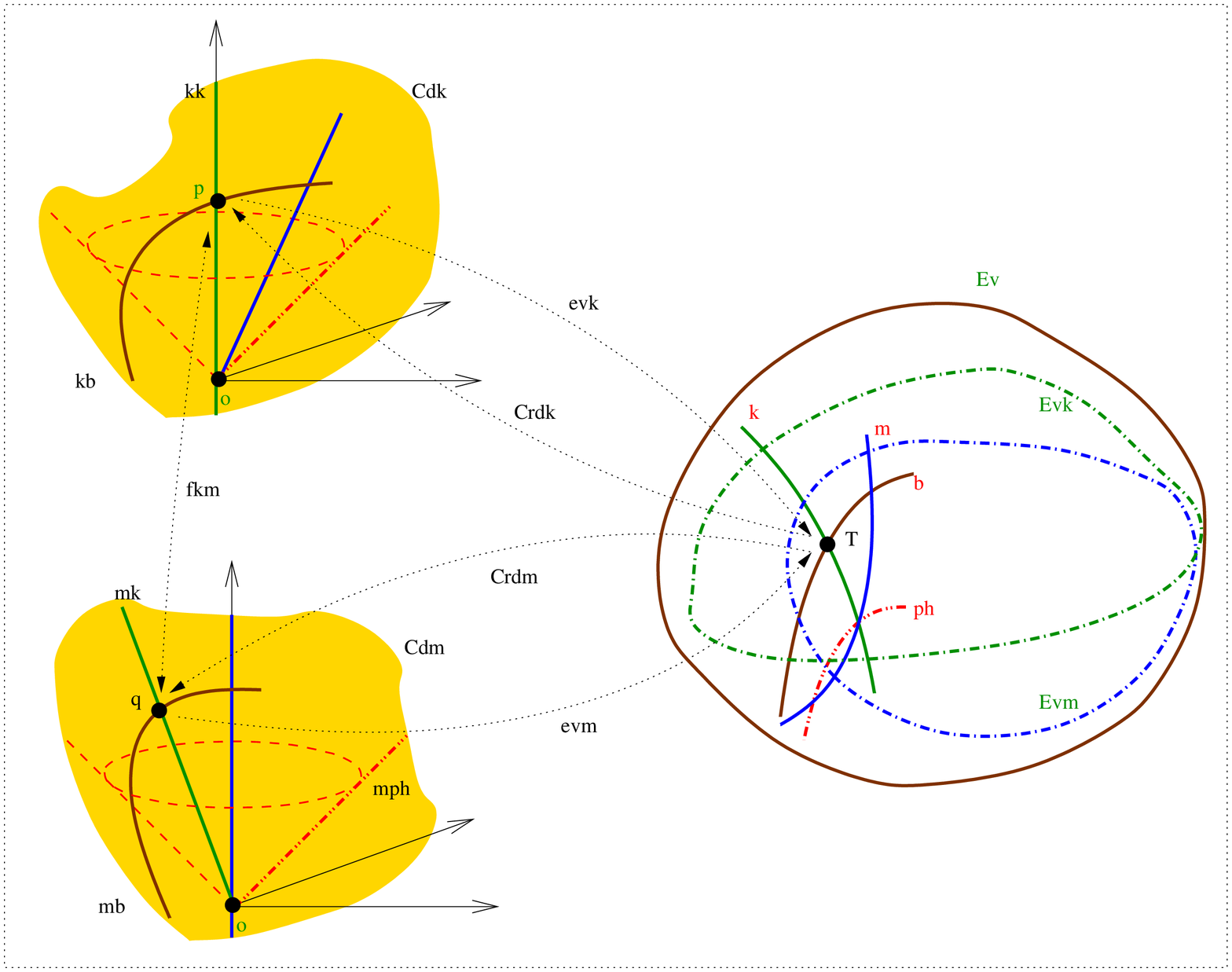}
\caption{\label{fig-fmk} Illustration of the basic definitions.}
\end{center}
\end{figure}

We need some definitions and notations to formulate our other axioms.
$\Q^d$ is called the \df{coordinate system} and its elements are referred to as \df{
coordinate points}.
We use the notations 
\begin{equation}
\Df{\vpp_\sigma}\leteq\langle p_2,\ldots, p_d\rangle \quad \text{ and }\quad \Df{p_\tau}\leteq p_1
\end{equation} 
for the \df{space component} and for the \df{time component} of $\vpp\in\Q^d$, respectively.
The \df{event} (the set of bodies) observed by observer $m$ at
coordinate point $\vpp$ is:
\begin{equation}
\Df{ev_m(\vpp)}\leteq\Setopen b\in \B \::\: \W(m,b,\vpp)\Setclose.
\end{equation}
The \df{coordinate-domain} of observer $m$ is the set of coordinate points where $m$
observes something:
\begin{equation}
\Df{Cd_m}\leteq\Setopen \vpp \in \Q^d\::\: ev_m(\vpp)\neq \emptyset \Setclose.
\end{equation}

Now we formulate our first axiom on observers.
Historically this natural axiom goes back to Galileo Galilei or even to d'Oresme of around 1350,
see, e.g., \cite[p.23, \S 5]{AMNsamp}, but it is very probably a prehistorical assumption, see remark below.
It simply states that each observer thinks that he rests in the origin of the space part of his coordinate system.

\begin{description}
\item[\Ax{AxSelf^-}] An observer observes himself at a coordinate point iff the space component of this point is the origin:
\begin{equation}
\forall m \in \Ob \enskip \forall \vpp\in Cd_m\quad \big(\,m\in ev_m(\vpp) \iff \vpp_\sigma=\vo\;\big).
\end{equation}
\end{description}

\begin{rem}
At first glance it is not clear why \ax{AxSelf} is so natural.
As an explanation, let us consider the following simple example.
Let us imagine that we are watching sunset.
What do we see? 
We do not see and feel that we are rotating with the Earth but that the Sun is moving towards the horizon;
and according to our (the Earth's) reference system we are absolutely right.
But we learned at primary school that ``the Earth rotates and goes around the Sun.''
So why does not this (that is, the adoption of the heliocentric system) mean that \ax{AxSelf} and our impression about the sunset above are simply wrong?
That is so, because the debate between geocentric and heliocentric systems was not about \ax{AxSelf}, but about how to choose the best observer (reference frame) if we want to study the motions of planets in our solar system.%
\footnote{Here we consider only the basic idea of the two systems (that is, whether the Earth or the Sun is stationary) and not their details (e.g., epicycles).
Of course, Ptolemy's geocentric model was wrong in its details since even if we fix the Earth as reference frame, the other planets will go around not the Earth but the Sun.
It is interesting to note that Tycho Brahe worked out the a correct geocentric system in which the Sun and the Moon move around the Earth and the other planets move around the Sun.} 
As reference frames, those of the Earth, the Sun, and even the Moon are equally good.
However, if we would like to calculate the motions of the planets, the Sun's is obviously the most convenient.
\end{rem}

Now we formulate our axiom about the constancy of the speed of photons.
For convenience, we choose $1$ for this speed.
\begin{description}
\item[\Ax{AxPh_0}] For every \emph{inertial} observer, there is a photon through two
coordinate points $\vpp$ and $\vqq$ iff the slope of $\vpp-\vqq$ is $1$:
\begin{equation}
\begin{split}
\forall m\in \IOb\enskip \forall \vpp,\vqq\in \Q^d\ \big(\,&
|\vpp_\sigma-\vq_\sigma|=|p_\tau-q_\tau| \iff \\ &\Ph\cap ev_m(\vpp)\cap
ev_m(\vqq)\ne\emptyset \,\big),
\end{split}
\end{equation}
\end{description}
where, the \df{Euclidean length} of $\vpp\in \Q^n$ is defined as $\Df{|\vpp|}\leteq\sqrt{\resizebox{!}{8pt}{$p_1^2+\ldots+p_n^2$}}$ 
for any $n\ge 1$.

This axiom is a well-known assumption of SR, see, e.g., \cite{logst}, \cite[\S 2.6]{d'Inverno}.

The set of nonempty events observed by observer $m$ is:
\begin{equation}
\Df{Ev_m}\leteq\Setopen ev_m(\vpp) \::\: ev_m(\vpp)\neq\emptyset\Setclose,
\end{equation}
and the set of all observed events is:
\begin{equation}
\Df{Ev}\leteq\Setopen e\in Ev_m \::\: m\in \Ob\Setclose.
\end{equation}

\begin{rem}
For convenience, we quantify over events too.
That does not mean that we abandon our first order language.
It is just a new abbreviation that simplifies the formalization of our axioms.
Instead of events we could speak about observers and spacetime locations.
For example, instead of $\forall e\in Ev_m\enskip \phi$ we could write $\forall \vpp\in Cd_m\enskip \phi[e\!\leadsto\! ev_m(\vpp)]$, where none of $p_1\ldots p_d$ occurs free in $\phi$, and $\phi[e\!\leadsto\! ev_m(\vpp)]$ is the formula achieved from $\phi$ by substituting $ev_m(\vpp)$ for $e$ in all occurrences.
Similarly, we can replace $e\in Ev_m$ by $\exists \vpp\in Cd_m\enskip e=ev_m(\vpp)$ and $\forall e\in Ev$ by $\forall m\in\Ob\enskip\forall e\in Ev_m$.
\end{rem}

By the next axiom we assume that {\em inertial} observers observe the same events.
\begin{description}
\item[\Ax{AxEv}] Any \emph{inertial} observer coordinatizes the same set of events:
\begin{equation}
\forall m,k\in \IOb \quad Ev_m=Ev_k.\footnotemark
\end{equation}
\end{description}
\footnotetext{Hint for translating this axiom into our first order language: replace $Ev_m=Ev_k$ by $\forall e\in Ev_m\enskip e\in Ev_k$ or see \cite{Twp}.}

We define the {\bf coordinate-function} of observer $m$, 
in symbols $Crd_m$, as the inverse of the event-function, that is, 
\begin{equation}
\Df{Crd_m}\leteq ev_m^{-1},
\end{equation}
where $R^{-1}:=\setopen\langle y,x\rangle : \langle x,y\rangle \in R\setclose$ is the first-order definition of the {\bf inverse} of binary relation $R$.
We note that the coordinate-functions are only binary relations by this definition, but one can easily prove from \ax{AxPh_0} that, if $m$ is an {\em inertial} observer, $Crd_m$ is a bijection from $Ev_m$ to $Cd_m$, see Proposition \ref{prop-tr} way below.

\begin{conv}\label{conv-crd}
Whenever we write $Crd_m(e)$, we mean that there is a unique $\vqq \in Cd_m$ such that $ev_m(\vqq)=e$, and this unique $\vqq$ is denoted by $Crd_m(e)$.
That is, when we talk about the value $Crd_m(e)$, we postulate that it exists and is unique.
\end{conv}

We say that events $e_1$ and $e_2$ are \df{simultaneous}
\label{sim} for observer $m$, in symbols $e_1\Df{\!\rule{0pt}{8pt}\!\sim_m\!} e_2$, iff $e_1$
and $e_2$ have the same time-coordinate in $m$'s coordinate-domain,
that is, if $Crd_m(e_1)_\tau=Crd_m(e_2)_\tau$.
To formulate time differences measured by observers, 
we use $\Df{\time_m}(e_1,e_2)$ as an
abbreviation for $|Crd_m(e_1)_\tau-Crd_m(e_2)_\tau|$, and we call it the
\df{elapsed time} between events $e_1$ and $e_2$ measured by observer $m$.
We note that $e_1\sim_m e_2$ iff $\time_m(e_1,e_2)=0$.
If $m\in e_1\cap e_2$, then
$\time_m(e_1,e_2)$ is called the {\em proper time} measured by $m$
between $e_1$ and $e_2$.
We use $\Df{\dist_m}(e_1,e_2)$ as an abbreviation
for $|Crd_m(e_1)_\sigma-Crd_m(e_2)_\sigma|$ and we call it the \df{spatial
distance} of events $e_1$ and $e_2$ according to observer $m$.
We note that when we write $\dist_m(e_1,e_2)$ or $\time_m(e_1,e_2)$, 
we assume that $e_1$ and $e_2$ have unique coordinates by
Convention~\ref{conv-crd}.

By the next axiom we assume that \emph{inertial} observers use the same units of measurement.

\begin{description}
\item[\Ax{AxSimDist}]
If events $e_1$ and $e_2$ are simultaneous for both \emph{inertial} observers $m$ and $k$, then $m$ and $k$ agree as for the spatial distance between $e_1$ and $e_2$:
\begin{equation}
\begin{split}
\forall m,k\in \IOb\enskip \forall e_1,e_2 \in Ev_m\cap Ev_k \quad \big(e_1\sim_m e_2\and e_1\sim_k e_2 & \\
\then\dist_m(e_1,e_2)=\dist_k(e_1,e_2)&\,\big).
\end{split}
\end{equation}
\end{description}

Let us collect these axioms in an axiom system:
\begin{equation}
\boxed{\ax{SpecRel_\d}\leteq\setopen \ax{AxSelf^-}, \ax{AxPh_0}, \ax{AxEv},\ax{AxSimDist} \setclose}
\end{equation}
Now for each natural number $d\ge2$, we have a first-order theory of SR.
Usually we omit the dimension parameter $d$.
From the few axioms
introduced so far, we can deduce the most frequently quoted
predictions of SR: 
\begin{itemize}
\item[(i)] ``moving clocks slow down,'' 
\item[(ii)] ``moving meter-rods shrink'' and 
\item[(iii)] ``moving pairs of clocks get out of synchronism.'' 
\end{itemize}
For more detail, see, for example, \cite{AMNsamp,pezsgo,logst}.

Obviously \ax{SpecRel} is too weak to answer any question about acceleration and hence about gravitation via EEP since \ax{AxSelf^-} is its only axiom that mentions non-inertial observers too.
To extend \ax{SpecRel}, we now formulate axioms about non-inertial observers called \df{accelerated observers}.

We assume the following very natural axiom for all observers.

\begin{description}
\item[\Ax{AxEvTr}] Whenever an observer participates in an event, he also coordinatizes this event:
\begin{equation}
\forall m\in \Ob\enskip \forall e\in Ev \quad \big(\,m\in e \then
e\in Ev_m\,\big).
\end{equation}
\end{description}
We note that \ax{AxEvTr} is not a consequence of \ax{SpecRel} even for {\em inertial} observers.

We also assume the following technical axiom:
\begin{description}
\item[\Ax{AxSelf^+}]
The set of time-instances in which an observer is present in its own
world-view is connected, that is,
\begin{equation}
\forall m\in\Ob\enskip \{ p_\tau : m\in ev_m(\vpp)\}\quad \mbox{is
connected,}
\end{equation}
\end{description}
\noindent 
where $I\subseteq \Q$ is said to be \df{connected} iff $(x,y)\subseteq I$ for all $x,y\in I$,
and the \df{interval} between $x,y\in \Q$ is defined as: 
\begin{equation}
\Df{(x,y)}\leteq\{z\in\Q:x<z<y \text{ or } y<z<x\}.
\end{equation}

To connect the coordinate-domains of the accelerated and the
inertial observers, we are going to formalize the statement that each accelerated observer, at
each moment of his life, coordinatizes the nearby
world for a short while as an \emph{inertial} observer.
First we introduce the relation of being a co-moving observer.
To do so, we define the (coordinate) \df{neighborhood} of event
$e$ with radius $\delta \in \Q^+$ according to observer $k$ as:
\begin{equation}
\Df{B^\delta_k(e)}\leteq\Setopen \vpp\in Cd_k \::\: \exists \vqq \in Cd_k \quad
ev_k(\vqq)=e \and |\vpp-\vqq|<\delta\Setclose.
\end{equation}
Observer $m$ is called a \df{co-moving observer} of observer
$k$ at event $e$, in symbols $\Df{m \com_e k}$, iff the following
holds:
\begin{equation}
\begin{split}
\forall \varepsilon \in \Q^+ \;\exists \delta \in \Q^+\enskip&\forall \vpp \in B^{\delta}_k(e)\\
&\left|\vpp-Crd_m\big(ev_k(\vpp)\big)\right| \leq\varepsilon\big|\vpp-Crd_k(e)\big|,
\end{split}
\end{equation}
where $\Df{\Q^+}$ denotes the set of \df{positive elements} of $\Q$, that is, 
\begin{equation}
\Q^+\leteq\setopen x\in \Q:0<x\setclose.
\end{equation}

\begin{rem}\label{rem-comove}
Note that $Crd_m(e)=Crd_k(e)$, and thus also $e\in Ev_m$ if $m
\com_e k$ and $e\in Ev_k$ [to see that let $\vpp=Crd_k(e)\in B^{\delta}_k(e)$].
Note also that $m\com_e k$ for any observer $m$ if $e\not\in Ev_k$ since $B^\delta_k(e)=\emptyset$ if $e\not \in Ev_k$ by definition.
\end{rem}

Behind the definition of co-moving observers is the following intuitive
image: as we zoom in the neighborhood of the
coordinate point of the given event, the world-views of the
two observers are getting more and more similar.
The following axiom gives the promised connection between the
world-views of the inertial and the accelerated observers:

\begin{description}
\item[\Ax{AxAcc}] At any event in which an observer coordinatizes himself, there is a co-moving \emph{inertial} observer:
\begin{equation}
\forall k \in \Ob \enskip \forall e \in Ev_k \quad (\,k\in e \then \exists m\in \IOb \enskip m \com_e k\,).
\end{equation}
\end{description}
\emph{Inertial} observer $m$ is called a {\bf co-moving inertial observer} of observer $k$ if there is an event $e\in Ev_k$ such that $k\in e$ and $m \com_e k$.

\begin{rem}\label{rem-axacc}
(1) From \ax{AxAcc} follows by Convention \ref{conv-crd}, that 
no observer can encounter an event more than once,
that is, if $k\in \Ob$, $e\in Ev$ and $\vpp,\vqq\in Cd_k$ such that $k\in e\in Ev_k$ and $e=ev_k(\vpp)=ev_k(\vqq)$, 
then $\vpp=\vqq$.
It is true since $Crd_k(e)$ is written in \ax{AxAcc}.

(2) From \ax{AxAcc} and \ax{AxEv} follows that 
any \emph{inertial} observer coordinatizes every event that an observer encounters,
 that is, if $m\in\IOb$, $k\in\Ob$ and $e\in Ev$ such that $k\in e\in Ev_k$, 
then there is a $\vpp \in Cd_m$ such that $ev_m(\vpp)=e$.
It is true since \emph{inertial} observers coordinatize the same events by \ax{AxEv} and $e\in Ev_m$ if $m\com_e k$ and $e\in Ev_k$, see Remark \ref{rem-comove}.
\end{rem}

Let us call the set of the axioms introduced so far \ax{AccRel^0_\d}:
\begin{equation}
\boxed{
\ax{AccRel^0_\d}\leteq\ax{SpecRel_\d}\cup\Setopen\ax{AxEvTr},\ax{AxSelf^+}, \ax{AxAcc}\Setclose}
\end{equation}

Surprisingly \ax{AccRel^0_\d} is not strong enough to prove
properties of accelerated clocks such as the Twin Paradox, see
Theorems 3.5 and 3.7 and Corollary 3.6 in \cite{Twp}.
The additional assumption we need is that every bounded non-empty subset of the
quantity part has a supremum.
It expresses a second-order logic property
(because it concerns all subsets) which we cannot use in a first-order axiom
system.
So instead of it we use a kind of ``induction'' axiom schema.
Let $\phi(x,\vy\,)$ be a first-order formula of our language.
\begin{description}
\item[\Ax{AxSup_\phi}] Every subset of $\Q$ definable by $\phi(x,\vy\,)$ with parameters $\vy$ has a supremum if it is non-empty and \df{bounded}.
\end{description}
\noindent A first-order formula expressing \ax{AxSup_\phi} can be found in
\cite{Twp}, or \cite{mythes}.
Our axiom scheme \ax{IND} below says that every
non-empty bounded subset of $\Q$ that is definable in our language
has a supremum:
\begin{equation}
\Ax{IND}\leteq\Setopen \ax{AxSup_\varphi}\::\: \varphi \text{ is a first-order formula of our language} \Setclose.
\end{equation}
Note that \ax{IND} is true in any model whose quantity part is the field of real numbers.
For more detail about \ax{IND}, see \cite{Twp, mythes}.

Let us call the set of the axioms introduced so far \ax{AccRel_\d}:
\begin{equation}
\boxed{
\ax{AccRel_\d}\leteq\ax{AccRel^0_\d}\cup\ax{IND}}
\end{equation}
We note that the Twin Paradox is provable in \ax{AccRel}, see \cite{Twp, mythes}.

\section{Gravitational time dilation}%
\label{thm-sec}

Let us go on to state our theorems about GTD.
Recall that GTD roughly says that ``gravitation makes time flow slower,'' that is to say, the clocks in the bottom of a tower run slower than the clocks in the top of the tower.
We use EEP to treat gravitation in \ax{AccRel}.
So instead of gravitation we will talk about acceleration and instead of towers we will talk about spaceships.
This way GTD becomes the following statement:
``the time in the back of an (uniformly) accelerated spaceship flows slower than in the front of the spaceship.'' 
Here we concentrate on the general case when the spaceship is not necessarily uniformly accelerated.
This case corresponds to the situation when the tower is in a possibly changing gravitational field.
Now let us begin to formulate this statement in our first-order language.

\begin{figure}[h!t]
\small
\begin{center}
\psfrag{a}[tl][tl]{$k$}
\psfrag{m}[l][l]{$m$}
\psfrag{e}[b][b]{$e$}
\psfrag{e1}[tl][tl]{$e_1$}
\psfrag{e'}[bl][bl]{$e'$}
\psfrag{e2}[br][br]{$e_2$}
\psfrag{2l}[l][l]{$2\lambda$}
\psfrag{l}[b][b]{$\lambda$}
\psfrag{ph1}[tr][tr]{$ph_1$}
\psfrag{ph2}[br][br]{$ph_2$}
\psfrag{text1}[tl][tl]{$(a)$}
\psfrag{text2}[tl][tl]{$(b)$}
\includegraphics[keepaspectratio, width=\textwidth]{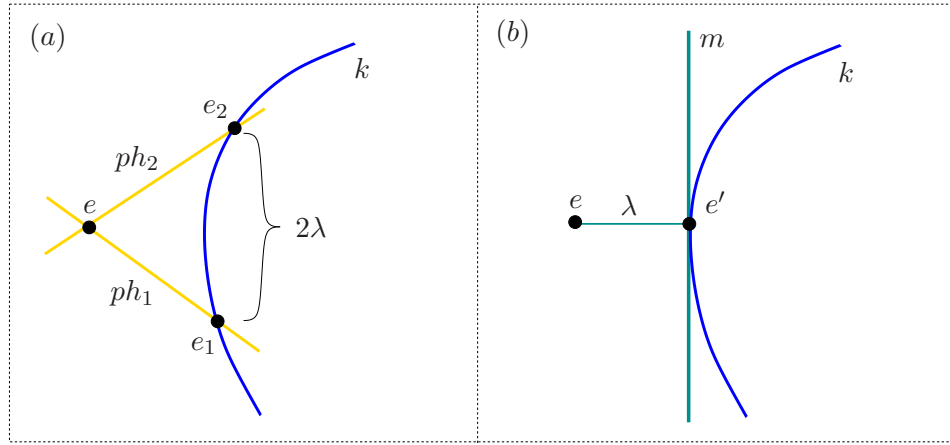}
\caption{\label{distfig} Illustrations of the radar distance and the Minkowski distance, respectively.}
\end{center}
\end{figure}

To talk about spaceships, we need a concept of distance between
events and observers.
We have two natural candidates for that:
\begin{itemize}
\item Event $e$ is at \df{radar distance} $\lambda\in\Q^+$ from observer $k$ iff there are events $e_1$ and $e_2$ and photons $ph_1$ and $ph_2$ such that $k\in e_1\cap e_2$, $ph_1\in e\cap e_1$, $ph_2\in e\cap e_2$ and $\time_k(e_1,e_2)=2\lambda$.
Event $e$ is at \df{radar distance} $0$ from observer $k$ iff $k\in e$.
See $(a)$ of Figure \ref{distfig}.
\item Event $e$ is at \df{Minkowski distance} $\lambda\in\Q$ from observer $k$ iff there is an event $e'$ such that $k\in e'$, $e\sim_m e'$ and $\dist_m(e,e')=\lambda$ for every co-moving inertial observer $m$ of $k$ at $e'$.
See $(b)$ of Figure \ref{distfig}.
\end{itemize}

We say that observer $k$ thinks that body $b$ is at constant radar distance from him iff the radar distance of every event in which $b$ participates is the same.
The notion of constant Minkowski distance is analogous.

The \df{world-line} of body $b$ according to observer $m$ is defined as
the set of the coordinate points where $b$ was observed by $m$:
\begin{equation}
\Df{wl_m(b)}\leteq\Setopen \vpp\in \Q^d \::\: b \in
ev_m(\vpp)\Setclose.
\end{equation}
To state that the {\em spaceship does not change its direction}, we need to introduce another concept.
We say that observers $k$ and $b$ are \df{coplanar} iff $wl_m(k)\cup wl_m(b)$ is a subset of a vertical plane in the coordinate system of an {\em inertial} observer $m$.
A plane is called a {\bf vertical plane} iff it is parallel with the time-axis.

Now we introduce two concepts of spaceship.
Observers $b$, $k$ and
$c$ form a \df{radar spaceship}, in symbols $\Df{\rship}$, iff $b$, $k$
and $c$ are coplanar and $k$ thinks that $b$ and $c$ are at constant
radar distances from him.
The definition of the \df{Minkowski spaceship}, 
in symbols $\Df{\mship}$, is analogous.

We say that event $e_1$ \df{precedes} event $e_2$
according to observer $k$ iff $Crd_m(e_1)_\tau\le Crd_m(e_2)_\tau$ for
all co-moving \emph{inertial} observers $m$ of $k$.
In this case, we also say that $e_2$ \df{succeeds} $e_1$ according to $k$.
We need these concepts to distinguish the past and the future light cones according to observers.
We note that since no time orientation is definable from our axiom system, we can only speak of orientation according to observers.

\begin{figure}[h!btp]
\small
\begin{center}
\psfrag{e}[r][r]{$e$}
\psfrag{he1}[l][l]{$\hat{e}_1$}
\psfrag{he2}[l][l]{$\hat{e}_2$}
\psfrag{hp1}[tr][tr]{$\hat{p}_1$}
\psfrag{hp2}[tl][tl]{$\hat{p}_2$}
\psfrag{te1}[br][br]{$\tilde{e}_1$}
\psfrag{te2}[l][l]{$\tilde{e}_2$}
\psfrag{tp1}[br][br]{$\tilde{p}_1$}
\psfrag{tp2}[bl][bl]{$\tilde{p}_2$}
\psfrag{e1}[r][r]{$e_1$}
\psfrag{e2}[lr][lr]{$e_2$}
\psfrag{ph1}[rt][rt]{$ph_1$}
\psfrag{ph2}[lt][lt]{$ph_2$}
\psfrag{m}[lb][lb]{$m$}
\psfrag{a}[tl][tl]{$k$}
\psfrag{l}[b][b]{$\lambda$}
\psfrag*{text1}[cb][cb]{(a)}
\psfrag*{text2}[cb][cb]{(b)}
\psfrag*{text3}[cb][cb]{(c)}
\includegraphics[keepaspectratio, width=\textwidth]{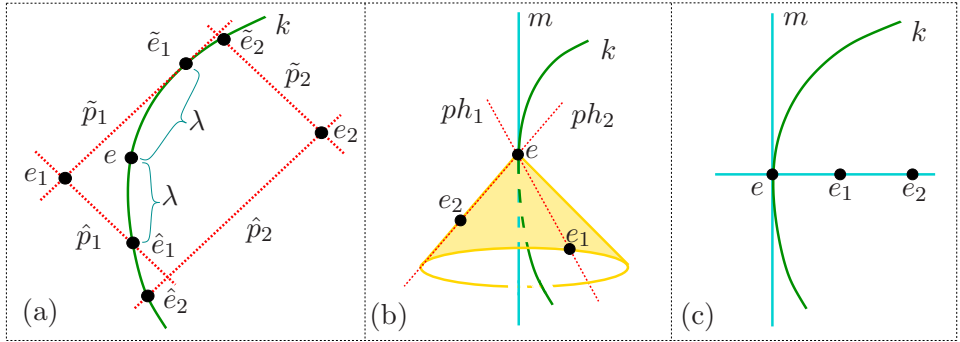}
\caption{\label{simfig} Illustrations of relations $e\simrad_k e'$, $e_1\simph_k e_2$ and $e_1\simmu_k e_2$, respectively.}
\end{center}
\end{figure}

We also need a concept to decide which events happened at the same
time according to an accelerated observer.
The following three
natural concepts offer themselves:
\begin{itemize}
\item Events $e_1$ and $e_2$ are \df{radar simultaneous} for observer $k$, in symbols $e_1\Df{\simrad_k} e_2$, 
iff there are events $e$, $\hat e_1$, $\hat e_2$, $\tilde e_1$, $\tilde e_2$ and photons $\tilde{p}_1$, $\tilde{p}_2$, $\hat{p}_1$, $\hat{p}_2$ 
such that $k\in e\cap \tilde{e}_i\cap\hat{e}_i$, $\hat{p}_i\in e_i\cap\hat{e}_i$, $\tilde{p}_i\in e_i\cap\tilde{e}_i$, ($\tilde{e}_i\neq\hat{e}_i$ or $e_i=e$) 
and $\time_k(e,\hat{e}_i)=\time_k(e,\tilde{e}_i)$ if $i\in\setopen 1,2\setclose$, see Figure \ref{simfig}.
\item Events $e_1$ and $e_2$ are \df{photon simultaneous} for observer $k$, in symbols $e_1\Df{\simph_k} e_2$, iff there are an event $e$ and photons $ph_1$ and $ph_2$ such that $k\in e$, $ph_1\in e\cap e_1$, $ph_2\in e\cap e_2$ and $e_1$ and $e_2$ precede $e$ according to $k$.
See $(b)$ of Figure \ref{simfig}.
\item Events $e_1$ and $e_2$ are \df{Minkowski simultaneous} for observer $k$, in symbols $e_1\Df{\simmu_k} e_2$, iff there is an event $e$ such that $k\in e$ and $e_1$ and $e_2$ are simultaneous for any co-moving \emph{inertial} observer of $k$ at $e$.
See $(c)$ of Figure \ref{simfig}.
\end{itemize}

We note that, for \emph{inertial} observers, the concepts of
radar simultaneity and Minkowski simultaneity coincide with the
concept of simultaneity introduced on page \pageref{sim}.

Radar simultaneity and Minkowski simultaneity are the two most natural generalisations (for non-inertial observers) of the standard simultaneity introduced by Einstein in \cite{Einstein}.
In the case of Minkowski simultaneity, 
the standard simultaneity of co-moving \emph{inertial} observers is rigidly copied, while in the case of radar simultaneity, the standard simultaneity is generalised in a more flexible way.
Dolby and Gull calculate and illustrate the radar simultaneity of some coplanar accelerated observers in \cite{Dolby-Gull}.
We note that $\simmu_k$ is an equivalence relation for observer $k$ iff $k$ does not accelerate.
So one can argue against regarding it as a simultaneity concept for non-inertial observers too.
However, we think that it is so straightforwardly generalised from the standard concept of simultaneity that it deserves to be forgiven for its weakness and to be called simultaneity.
The concept of photon simultaneity is the least usual and the most naive.
It is based on the simple idea that an event is happening right now iff it is seen to be happening right now.
Some authors require from a simultaneity concept to be an equivalence relation such that its equivalence classes are smooth spacelike hypersurfaces, see, e.g., Matolcsi~\cite{Matolcsi}.
In spite of the fact that equivalence classes of $\simph_k$ are neither smooth nor spacelike, we think that it deserves to be called simultaneity since it fulfills the most basic requirement that one may expect of a concept of simultaneity, see, e.g., Hogarth \cite{Hogarth} and Malament \cite{Malament}.
Moreover, this concept appears as a possible simultaneity concept in some of the papers investigating the question of conventionality/definability of simultaneity, see, e.g., Ben-Yami \cite{Ben-Yami}, Rynasiewicz \cite{Rynasiewicz}, Sarkar and Stachel \cite{Sarkar-Stachel}.
We also note that all of the introduced simultaneity and distance concepts are experimental ones, that is, they can be determined by observers by the means of experiments with clocks and photons.

We distinguish the front and the back of the spaceship by the
direction of the acceleration, so we need a concept for direction.
We say that the \df{directions of $\vp\in \Q^d$ and $\vq\in \Q^d$ are the same}, 
in symbols $\vpp\Df{\upp}\vqq$, if $\vp$ and $\vq$ are spacelike vectors, and there is a $\lambda \in
\Q^+$ such that $\lambda \vp_\sigma=\vq_\sigma$, see $(a)$ of
Figure~\ref{figupp}.
When $\vp$ and $\vq$ are timelike vectors, we also use this notation if $p_\tau q_\tau>0$.
Spacetime vector $\vr$ is called \df{spacelike} iff $|\vr_\sigma|>|r_\tau|$, 
\df{lightlike} iff $|\vr_\sigma|=|r_\tau|$, and
\df{timelike} iff $|\vr_\sigma|<|r_\tau|$.

\begin{figure}[h!btp]
\small
\begin{center}
\psfrag{p}[lb][lb]{$\vpp$} 
\psfrag{ps}[lb][lb]{$\vpp_\sigma$}
\psfrag{qs}[lb][lb]{$\vq_\sigma$} 
\psfrag{ph}[rb][rb]{$ph$}
\psfrag{q}[r][r]{$\vqq$} 
\psfrag{q3}[l][l]{$\vq_3$}
\psfrag{q2}[l][l]{$\vq_2$} 
\psfrag{q1}[l][l]{$\vq_1$}
\psfrag{aa}[l][l]{$(a)$} 
\psfrag{bb}[l][l]{$(b)$}
\psfrag{o}[t][t]{$\vo$} 
\psfrag{b}[t][t]{$b$} 
\psfrag{k}[t][t]{$c$}
\psfrag{e}[t][t]{$e$} 
\psfrag{eb}[t][t]{$e_b$}
\psfrag{ek}[t][t]{$e_c$} 
\psfrag{k1}[t][t]{$c'$}
\psfrag{b1}[t][t]{$b'$}
\includegraphics[keepaspectratio, width=\textwidth]{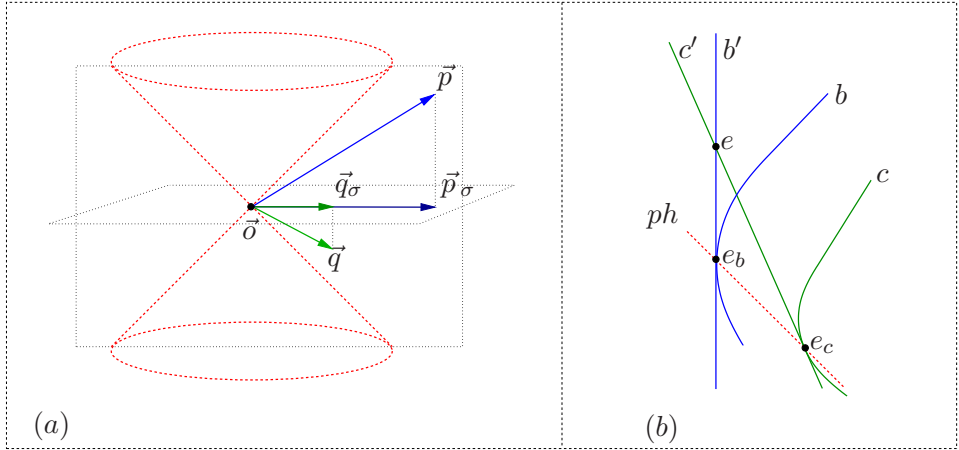}
\caption{\label{figupp} $(a)$ illustrates $\vpp\upp\vqq$, and
$(b)$ illustrates that observer $c$ is approaching observer $b$,
as seen by $b$ with photons.}
\end{center}
\end{figure}

Now let us focus on the definition of acceleration in our FOL setting.
We define the \df{life-curve} of observer $k$ according to observer
$m$ as the life-line of $k$ according to $m$ {\em parametrized by
the time measured by $k$}, formally:
\begin{equation}
\begin{split}
&\Df{lc^k_m}\leteq\Setopen \langle t,\vpp \rangle\in \Q\times Cd_m \::\right.\\
&\qquad\left.\exists \vqq\in Cd_k \quad k\in ev_k(\vqq)=ev_m(\vpp)\and q_\tau=t\Setclose.
\end{split}
\end{equation}
For the most important properties of this concept, see Proposition \ref{prop-tr} in Section \ref{thm-sec}.
The life-curves of observers and the derivative $f'$ of a given function $f$ are both
first-order definable concepts, see \cite{Twp, mythes}.
Thus if the life-curve of observer $k$ according to observer $m$ is a function, then the
following definitions are also first-order ones. 
The \df{relative velocity} $\Dff{\vvkm}$ of observer $k$ according to observer $m$ at instant
$t\in\Q$ is the derivative of the life-curve of $k$ according to $m$
at $t$ if it is differentiable at $t$ and undefined otherwise.
The \df{relative acceleration} $\Dff{\vakm}$ of
observer $k$ according to observer $m$ at instant $t\in\Q$ is the
derivative of the relative velocity of $k$ according to $m$ at $t$ if it is differentiable
at $t$ and undefined otherwise.

Spacetime vectors $\vpp$ and $\vqq$ are called 
\df{spacelike-separated}, in symbols $\vpp\Df{\seq}\vqq$, iff $\vp-\vq$ is a spacelike vector;
\df{lightlike-separated}, in symbols $\vpp\Df{\pheq}\vqq$, iff $\vp-\vq$ is a lightlike vector;
\df{timelike-separated}, in symbols $\vpp\Df{\teq}\vqq$, iff $\vp-\vq$ is a timelike vector.
Events $e_1$ and $e_2$ are called spacelike-separated (lightlike-separated; timelike-separated), in
symbols $e_1\seq e_2$ ($e_1\pheq e_2$; $e_1\teq e_2$), iff $Crd_m(e_1)$ and $Crd_m(e_2)$ are such for any {\em inertial} observer $m$.

We say that \df{the direction of the spaceship $\ship$ is the same as that of the} 
\df{acceleration of $k$} iff the following holds:
\begin{equation}
\begin{split}
&\forall m \in \IOb \enskip \forall t \in \dom \vakm\enskip \enskip \forall \vpp,\vqq \in Cd_m\quad \big(\,c\in ev_m(\vpp) \and \\
&\qquad\qquad\quad\enskip b\in ev_m(\vqq)\and \vpp\seq \vqq \then\vakm(t)\upp (\vpp-\vqq)\,\big),
\end{split}
\end{equation}
where $\Df{\dom} R\leteq\setopen x : \exists y \enskip \langle x,y\rangle \in R\setclose$ is the first-order definition of the \df{domain} of binary relation $R$.

The (signed) \df{Minkowski length} of $\vpp\in \Q^d$ is
\begin{equation}
\Df{\mu(\vpp)}\leteq\left\{
\begin{array}{rl}
\sqrt{\rule{0pt}{11pt} p_\tau^2-|\vpp_\sigma|^2} & \text{ if } p_\tau^2\ge|\vpp_\sigma|^2, \\
-\sqrt{\rule{0pt}{11pt}|\vpp_\sigma|^2-p_\tau^2} & \text{ otherwise, }
\end{array}
\right.
\end{equation}
and the \df{Minkowski distance} between $\vpp$ and $\vqq$ is
$\mu(\vpp,\vqq)\leteq\mu(\vpp-\vqq)$.
We use the signed version of the Minkowski length because it contains
two kinds of information: 
(i) the length of $\vpp$, and (ii) whether it is spacelike, lightlike or timelike.
Since the length is always
non-negative, we can use the sign of $\mu(\vpp)$ to code (ii).

The \df{acceleration} of observer $k$ at instant $t\in\Q$ is defined as the unsigned Minkowski length of the relative acceleration according to any {\em inertial} observer $m$ at $t$, that is,
\begin{equation}
\Df{a_k(t)}\leteq-\mu\big(\vakm(t)\big).
\end{equation}
The reason for the ``$-$'' sign in this definition is the fact that $\mu\big(\vakm(t)\big)$ is negative 
since $\vakm(t)$ is a spacelike vector, see Propositions \ref{prop-wellpar} and \ref{prop-vmorta}.
The acceleration is a well-defined concept since it is independent of the choice of the inertial observer $m$, see Theorem \ref{thm-poi} and Proposition \ref{prop-inv}.
We say that observer $k$ is \df{positively accelerated} iff $a_k(t)$ is defined and greater than $0$ for all $t\in \dom lc^k_k$.
Observer $k$ is called \df{uniformly accelerated} iff there is an $a\in\Q^+$ such that $a_k(t)=a$ for all $t\in \dom lc^k_k$.

We say that \df{the clock of $b$ runs slower than the clock of $c$ as} 
\df{seen by \,$k$\, by radar} iff
$\time_b(e_b,\bar{e}_b)<\time_c(e_c,\bar{e}_c)$ for all events $e_b,\bar{e}_b, e_c,
\bar{e}_c$ for which $b\in e_b\cap \bar{e}_b$, $c\in e_c\cap \bar{e}_c$ and
$e_b\simrad_k e_c$, $\bar{e}_b\simrad_k \bar{e}_c$.
If it is \df{seen by photons}, we use $\simph_k$ instead of $\simrad_k$.
Similarly, if it is \df{seen by Minkowski simultaneity}, we use $\simmu_k$ instead of $\simrad_k$.

The following theorem will show that the flow of time as seen by photons is strongly connected with the following two concepts.
We say that observer $c$ is \df{approaching} (or \df{moving away} from) observer $b$ as seen by $b$ by photons at event $e_b$ iff $b\in e_b$ and, for all events
$e_c$ if $c\in e_c$ and $e_b\simph_b e_c$, there is an event $e$ such that $b',c'\in e$ for all co-moving {\em inertial} observers $b'$ and $c'$ of $b$ and $c$ at events $e_b$ and $e_c$, respectively, and $e_b$ precedes (succeeds) $e$ according to $b$, see (b) of Figure \ref{figupp}.
We say that $c$ is approaching (moving away from) $b$ as seen by $b$ by photons iff it is so for every event $e_b$ for which $b\in e_b$.
The idea behind these definitions is the following: 
two observers are considered approaching when they would meet if they have stopped accelerating at simultaneous events.

\begin{rem} We note that coplanar {\em inertial} observers seen by photons are approaching each other before the event of meeting and moving away from each other after it.
This fact explains the names of these concepts.
\end{rem}
 
\begin{rem}
There is no direct connection between the two concepts we have just introduced.
For example, there are easily constructable models of \ax{AccRel} and (uniformly accelerated) observers $b$ and $c$ such that $c$ is approaching $b$ seen by $b$ by photons while $b$ is moving away from $c$ seen by $c$ by photons, see the proof of Theorem \ref{thm-ob}.
\end{rem}

\begin{thm} \label{thm-ph}
Let $d\ge 3$.
Assume \ax{AccRel_\d}.
Let $b,c\in\Ob$ such that $c$ and $b$ are coplanar.
Then
\begin{itemize}
\item[(1)] If $c$ is approaching $b$ as seen by $b$ by photons, the clock of $b$ runs slower
than the clock of $c$ as seen by $b$ by photons.
\item[(2)] If $c$ is moving away from $b$ as seen by $b$ by photons, the clock of $c$ runs
slower than the clock of $b$ as seen by $b$ by photons.
\end{itemize}
\end{thm}

Now let us state a theorem about the clock-slowing effect of gravitation in the radar spaceship:

\begin{thm} \label{thm-rad}
Let $d\ge 3$.
Assume \ax{AccRel_\d}.
Let $\rship$ be a radar spaceship such that:
\begin{itemize}
\item[(i)] $k$ is positively accelerated,
\item[(ii)] the direction of the spaceship is the same as that of the acceleration of
$k$.
\end{itemize}
Then both (1) and (2) hold:
\begin{itemize}
\item[$(1)$] The clock of $b$ runs slower than the clock of $c$ as seen by $k$ by radar.
\item[$(2)$] The clock of $b$ runs slower than the clock of $c$ as seen by each of $k$, $b$ and $c$ by photons.
\end{itemize}
\end{thm}

To state a similar theorem for Minkowski spaceships, we need the
following concept.
We say that observer \df{$b$ is not too far behind} 
positively accelerated observer $k$ iff the following holds:
\begin{equation}
\begin{split}
\forall m \in \IOb \enskip \forall t \in \dom \vakm\enskip \forall \vpp,\vqq \in Cd_m \quad \big(\,k\in ev_m(\vpp)\and\quad& \\
b\in ev_m(\vqq)\and ev_m(\vpp)\simmu_k ev_m(\vqq) \and \vakm(t)\upp (\vpp-\vqq) \then& \\
\qquad\qquad\quad\forall \tau \in \dom\vakm \quad -\mu(\vpp,\vqq)<1/a_k(\tau)\,&\big).
\end{split}
\end{equation}
Now we can state the theorem about the clock-slowing effect of gravitation in the Minkowski spaceship:

\begin{thm}\label{thm-mu}
Let $d\ge 3$.
Assume \ax{AccRel_\d}.
Let $\mship$ be a Minkowski spaceship such that:
\begin{itemize}
\item[(i)] $k$ is positively accelerated,
\item[(ii)] the direction of the spaceship is the same as that of the acceleration of $k$,
\item[(iii)] $b$ is not too far behind $k$.
\end{itemize}
Then both (1) and (2) hold:
\begin{enumerate}
\item The clock of $b$ runs slower than the clock of $c$ as
seen by $k$ by Minkowski simultaneity.
\item The clock of $b$ runs slower than the clock of $c$ as seen by each
of $k$, $b$ and $c$ by photons.
\end{enumerate}
\end{thm}

We have seen that gravitation (acceleration) makes ``time flow slowly.'' 
However, we left the question open which feature of gravitation (its ``magnitude'' or its ``direction'') plays role in this effect.
The following theorem shows that two observers, say $b$ and $c$, can feel the same gravitation while the clock of $b$ runs slower than the clock of $c$.
Thus it is not the ``magnitude'' of the gravitation that makes ``time flow more slowly.''

\begin{thm} \label{thm-ob}
Let $d\ge 3$.
Then there is a model of \ax{AccRel_\d} and there
are $b,c\in \Ob$ such that
$a_b(t)=a_c(t)=1$ for all $t\in\Q$, but the clock of $b$ runs slower
than the clock of $c$ as seen by $b$ by photons (or
by radar or by Minkowski simultaneity).
\end{thm}

\section{Proofs of the Theorems}%
\label{proofs-sec}

Our next step is to prove the theorems introduced above.
First we have to develop the necessary tools.
Let us recall some basic first-order definable notations.
The {\bf composition} of binary relations $R$ and $S$ is defined as: 
\begin{equation}\Df{R \circ S}\leteq\Setopen \langle a,c\rangle \::\: \exists b \enskip \langle a,b\rangle \in R\and \langle b,c\rangle \in S \Setclose.\end{equation} 
The \df{range} of a binary relation $R$ is defined as $\ran R\leteq\setopen y:\exists x \enskip R(x,y)\setclose$.
\begin{rem}
We think of a {\bf function} as a special binary relation.
Notation $\Df{f:A\rightarrow B}$ denotes that $f$ is a function from $A$ to $B$, that is, $Dom f=A$ and $\ran f\subseteq B$.
Note that if $f$ and $g$ are functions, then 
\begin{equation}
(f \circ g) (x)=g\big(f(x)\big)
\end{equation}
for all $x\in \dom(f\circ g)$.
Notation $\Dff{f:A\parrow B}$ denotes that $f$ is a function, $\dom f\subseteq A$ and $\ran f\subseteq B$.
\end{rem}

We find that studying the relationships between the world-views is more illuminating than studying the world-views in themselves. 
Therefore, the following definition is a fundamental one.
The {\bf world-view transformation} between the coordinate-domains of observers $k$ and $m$ is defined as:
\begin{equation}
\Df{w^k_m} \leteq\Setopen\langle \vqq,\vpp\rangle \in Cd_k \times Cd_m\::\:ev_k(\vqq)=ev_m(\vpp)\Setclose,
\end{equation}
see Figure \ref{fig-fmk}.
We note that $w^k_m=ev_k\circ Crd_m$.
We also note that although the world-view transformations are only binary relations by this definition, axiom \ax{AxPh_0} turns them into functions, see \eqref{item-crd} in Proposition~\ref{prop-tr} below.

\begin{conv}\label{fmkconv}
Whenever we write $w^k_m(\vqq) $, we mean that there is a unique
$\vpp\in \Q^d$ such that $\langle \vqq,\vpp\rangle \in w^k_m$, and this $\vp$ is denoted by $w^k_m(\vqq)$.
\end{conv}

A map $f:\Q^d\rightarrow \Q^d$ is called {\bf Poincar\'e transformation} iff it is an affine transformation preserving the Minkowski distance, that is, $\mu\big(f(\vpp),f(\vqq)\big)=\mu(p,q)$.
A map $\tilde\varphi:\Q^d\rightarrow \Q^d$ is called {\bf field-automorphism induced map} iff there is an automorphism $\varphi$ of the field $\langle\Q,\cdot,+\rangle$ such that $\tilde\varphi(\vpp)=\langle \varphi(p_1),\ldots,\varphi(p_d)\rangle$ for every $\vp\in\Q^d$.

\begin{thm}
\label{thm-poi}
Let $d\ge 3$.
Let $m,k\in\IOb$.
Then
\begin{enumerate}
\item If \ax{AxPh_0} and \ax{AxEv} are assumed, $w^k_m$ is a Poincar\'e transformation composed by a dilation $D$ and a field-automorphism induced map $\tilde\varphi$.
\item If \ax{AxPh_0}, \ax{AxEv} and \ax{AxSimDist} are assumed, $w^k_m$ is a Poincar\'e transformation.
\end{enumerate}
\end{thm}

\begin{proof}[\textcolor{proofcolor}{On the proof}]
By \ax{AxPh_0} and \ax{AxEv}, $w^k_m$ is a bijection from $\Q^d$ to $\Q^d$ that preserves lines of slope 1, see \eqref{item-fmk} in Proposition \ref{prop-tr}.
Hence Item (1) is a consequence of the Alexandrov-Zeeman theorem generalized for fields, see, e.g., \cite{vroegindewey}, \cite{VKK}.

Now let us see why Item (2) is true.
By Item (1), it is easy to see that there is a line $l$ such that both $l$ and its $w^k_m$ image is orthogonal to the time-axis.
Thus by \ax{AxSimDist}, $w^k_m$ restricted to $l$ is distance preserving.
Consequently, both dilation $D$ and field-automorphism induced map $\tilde\varphi$ in Item (1) has to be the identity map.
Hence $w^k_m$ is a Poincar\'e transformation.
\end{proof}

\begin{prop}
\label{prop-tr}
Let $m,k,h\in \Ob.$ 
Then
\begin{enumerate}
\item \label{item-trfiff} $lc^k_m$ is a function iff (i) event $e$ has a unique coordinate in $Cd_m$ if $k\in e\in Ev_m\cap Ev_k$, and (ii) $ev_k(\vqq)=ev_k(\vqq')$ if $\vqq,\vqq'\in Cd_k$ such that $ev_k(\vqq),ev_k(\vqq')\in Ev_m$, $k\in ev_k(\vqq)\cap ev_k(\vqq')$ and $q_\tau=q'_\tau$.
\item \label{item-crd} Assume \ax{AxPh_0}, and let $m\in\IOb$.
Then $Cd_m=\Q^d$ and $ev_m$ is injective.
Thus $Crd_m$ and $w^k_m$ are functions.
\item \label{item-trfunct} Assume \ax{AxPh_0}, \ax{AxSelf^-}, and let $m\in\IOb$.
Then $lc^k_m$ is a function.
\item \label{item-fmk} Assume \ax{AxPh_0}, \ax{AxEv}, and let $m,k\in\IOb$.
Then $w^k_m$ is a bijection from $\Q^d$ to $\Q^d$.
\item \label{item-trkk} Assume \ax{AxSelf^-} and \ax{AxAcc}.
Then $\langle t,\vpp\rangle \in lc^k_k$ iff $\vpp\in Cd_k$, $\vpp_\sigma=\vo$ and $p_\tau=t$.
\item \label{item-tr} $lc^h_m\supseteq lc^h_k\circ w^k_m$ always holds and if we assume $\ax{\mathit{Ev_m\subseteq Ev_k}}$, then $lc^h_m= lc^h_k\circ w^k_m$.
\item \label{item-domtr} $\setopen q_\tau : k\in ev_m(\vqq)\setclose=\dom lc^k_k\supseteq \dom lc^k_m$ always holds and, if we assume \ax{AxAcc} and $m\in \IOb$, then $\dom lc^k_m=\dom lc^k_k$.
\item \label{item-rantr} $\ran lc^k_m\subseteq wl_m(k)$ always holds and if we assume \ax{AxEvTr}, then $\ran lc^k_m=wl_m(k)$.
\end{enumerate}
\end{prop}

\begin{proof}
Item \eqref{item-trfiff} is a straightforward consequence of the definition of $lc^k_m$.
To see that, let $R\leteq\setopen \langle t,\vqq\rangle\in \Q\times Cd_k \::\: k\in ev_k(\vqq) \and q_\tau=t\setclose$.
Then $lc^k_m=R\circ w^k_m=R\circ ev_k\circ Crd_m$.
Since $ev_k$ is a function and $Crd_m$ is an inverse of a function, it is easy to see that $lc^k_m$ is a function iff $Crd_m$ is a function on $\ran (R\circ ev_k)$ and $R\circ ev_k$ is a function to $\dom Crd_m=Ev_m$.
It is clear that $Crd_m$ is a function on $\ran (R\circ ev_k)$ iff (i) holds;
and it is also clear that $R$ is a function to $\dom Crd_m=Ev_m$ iff (ii) holds.
Hence $lc^k_m$ is a function iff both (i) and (ii) hold.

To prove Item \eqref{item-crd}, let $\vpp\in\Q^d$.
Then by \ax{AxPh_0}, there is a $ph\in Ph$ such that $ph\in ev_m(\vpp)$.
Hence $Cd_m=\Q^d$.
Moreover, if $\vqq \in \Q^d$ and $\vqq\neq\vpp$, then it is possible to choose this $ph$ such that $ph\not\in ev_m(\vqq)$ also holds.
Thus $ev_m$ is injective.
Consequently, both $Crd_m\leteq ev_m^{-1}$ and $w^k_m\leteq ev_k\circ Crd_m$ are functions.

To prove Item \eqref{item-trfunct}, we should check (i) and (ii) of Item \eqref{item-trfiff}.
By Item \eqref{item-crd}, (i) is true.
By \ax{AxSelf_0} if $k\in ev_k(\vqq)\cap ev_k(\vqq')$, then $\vq_\sigma=\vo=\vqq'_\sigma$.
Thus if $q_\tau=q'_\tau$ also holds, then $\vqq=\vqq'$.
Hence (ii) is also true.

Let us now prove Item \eqref{item-fmk}.
By Item \eqref{item-crd}, we already have that $ev_k$ is a bijection from $Cd_k$ to $Ev_k$ and $Crd_m$ is a bijection from $Ev_m$ to $Cd_m$.
By $\ax{AxEv}$, $Ev_k=Ev_m$.
Thus $w^k_m=ev_k\circ Crd_m$ is a bijection from $Cd_k$ to $Cd_m$.
But by Item \eqref{item-crd}, we also have that $Cd_k=Cd_m=\Q^d$.
Hence $w^k_m$ is a bijection from $\Q^d$ to $\Q^d$.

To prove Item \eqref{item-trkk}, let $\langle t,\vpp\rangle \in lc^k_k$.
Then $\exists\vqq\in Cd_k$ such that $k\in ev_k(\vqq)=ev_k(\vpp)$, $q_\tau=t$.
By \ax{AxSelf^-}, we have that $\vpp_\sigma=\vo$.
By \ax{AxAcc}, we have that $\vpp=\vqq$, see (1) of Remark \ref{rem-axacc}.
Hence $\vpp\in Cd_k$, $\vpp_\sigma=\vo$ and $p_\tau=t$.
The converse is also clear since if $\vpp\in Cd_k$ and $\vpp_\sigma=\vo$, then $k\in ev_k(\vpp)$ by \ax{AxSelf^-}.

To prove Item \eqref{item-tr}, let $\langle t,\vpp\rangle\in lc^h_k\circ w^k_m$.
That means $\exists\vc\in Cd_k$ such that $\langle t,\vc\,\rangle\in lc^h_k$ and $\langle\vc,\vpp\rangle\in w^k_m$, which is equivalent with $\exists\vqq\in Cd_h$ such that $h\in ev_h(\vqq)=ev_k(\vc\,)$, $q_\tau=t$ and $ev_k(\vc\,)=ev_m(\vpp)$.
Thus $\langle t,\vpp\rangle\in lc^h_m$.
To prove the converse inclusion, let $\langle t,\vpp\rangle\in lc^h_m$.
That means that there is a $\vqq\in Cd_h$ such that $h\in ev_h(\vqq)=ev_m(\vpp)$ and $q_\tau=t$.
By the assumption $Ev_m\subseteq Ev_k$, we have that $\exists \vc\in Cd_k$ such that $ev_k(\vc\,)=ev_m(\vpp)$.
Thus $\langle t,\vpp\rangle\in lc^h_k\circ w^k_m$.
That proves Item \eqref{item-tr}.

To prove Item \eqref{item-domtr}, let us recall that $t\in \dom lc^k_m$ iff there are $\vpp\in Cd_m$ and $\vqq\in Cd_k$ such that $k\in ev_m(\vpp)=ev_k(\vqq)$ and $q_\tau=t$.
From that, it easily follows that $t\in \dom lc^k_k$ iff there is a $\vqq\in Cd_k$ such that $q_\tau=t$ and $k\in ev_k(\vqq)$.
Thus $\setopen q_\tau : k\in ev_m(\vqq)\setclose=\dom lc^k_k\supseteq \dom lc^k_m$ is clear; 
and if we assume \ax{AxAcc} and $m\in\IOb$, then $\dom lc^k_k\subseteq \dom lc^k_m$ is also clear since \emph{inertial} observers coordinatize every event encountered by observers, see Remark \ref{rem-axacc}.
 
To prove Item \eqref{item-rantr}, recall that $\vpp\in\ran lc^k_m$ iff $\vp\in Cd_m$ and there are $t\in \Q$ and $\vqq\in Cd_k$ such that $k\in ev_m(\vpp)=ev_k(\vqq)$ and $q_\tau=t$.
Thus $\ran lc^k_m\subseteq wl_m(k)\leteq\setopen \vpp\in Cd_m \::\: k\in ev_m(\vpp)\setclose$ is clear.
If $\vpp\in wl_m(k)$, then $k\in ev_m(\vpp)$.
Therefore, by \ax{AxEvTr}, we have that $ev_m(\vpp)\in Ev_k$.
Thus there is a $\vqq\in Cd_k$ such that $ev_m(\vpp)=ev_k(\vqq)$.
Hence $\ran lc^k_m=wl_m(k)$.
\end{proof}

We say that a function $\gamma:\Q\parrow \Q^d$ is a \df{curve} if $\dom\gamma$ is connected.
A curve $\gamma$ is called \df{timelike} iff it is differentiable, and $\gamma'(t)$ is timelike for all $t\in \dom \gamma$.
We call a timelike curve $\alpha$ \df{well-parametrized} if $\mu\big(\alpha'(t)\big)=1$ for all $t\in \dom \alpha$.

\begin{prop}
\label{prop-wellpar} 
Let $d\ge 3$.
Assume \ax{AccRel^0_\d}.
Let $k\in \Ob$ and $m\in\IOb$.
Then $lc^k_m$ is a definable and well-parametrized timelike curve.
\end{prop}

\begin{proof}
It is clear that $lc^k_m$ is definable.
By \eqref{item-trfunct} in Proposition \ref{prop-tr} we have that $lc^k_m$ is a function.
To prove that $lc^k_m$ is also a curve, we need to show that $\dom lc^k_m$ is connected.
It is so because by Item \eqref{item-domtr} in Proposition \ref{prop-tr}, $\dom lc^k_m=\setopen p_\tau : k\in ev_m(\vpp)\setclose$ and the latter is connected by \ax{AxSelf^+}.
Hence $lc^k_m$ is a curve.

To complete the proof, we have to show that $lc^k_m$ is also timelike and well-parametrized.
To see that Proposition 5.2 in \cite{Twp} says that $lc^k_m$ is timelike and well-parametrized in the models of \ax{AccRel^0_\d}, we have to consider three things: 
1.\ Even if \ax{SpecRel_\d} that is used here is weaker than what was used in \cite{Twp}, it is still strong enough to prove that the world-view transformations are Poincar\'{e} transformations (see Theorem \ref{thm-poi}) and we used \ax{SpecRel_\d} only by this consequence in \cite{Twp}.
2.\ Even if \ax{AxAcc} is formulated in a different manner here, it is the same assumption as in \cite{Twp} if \ax{AxSelf^-} is assumed.
3.\ If we assume \ax{AxSelf^-} and $m\in\IOb$, then $lc^k_m$ defined here is the same as in \cite{Twp}.
All the three are easily verifiable, which is left to the reader.
\end{proof}

By \ax{IND}, a certain fragment of real analysis can be generalised for ordered fields and \emph{definable} functions within FOL.
The generalisations used herein can be found in Section \ref{lem-sec}.
For further details, see \cite{Twp, mythes}.
We refer to these generalisations by marking them ``\ax{IND}-'' iff they can be proved by but not without \ax{IND}.
The FOL generalisations of some theorems such as the chain rule can be proved without \ax{IND},
so they are naturally referred to without the ``\ax{IND}-'' mark.

Coordinate-points $\vpp$ and $\vqq$ are called \df{Minkowski orthogonal}, in symbols $\vpp\Df{\mort}\vqq$, iff $p_\tau\cdot q_\tau=p_2\cdot q_2+\ldots+p_d\cdot q_d$.

Let us introduce the following notation: 
\begin{equation}
\Df{dw^k_m}(\vpp)\leteq w^k_m(\vpp)-w^k_m(\vo\,).
\end{equation}

\begin{prop}
\label{prop-inv} 
Let $d\ge3$.
Assume \ax{SpecRel_\d}.
Let $m,k\in\IOb$ and $h\in \Ob$.
Then
\begin{enumerate}
\item \label{item-mort} $\vpp\mort\vqq$ iff $dw^k_m(\vpp)\mort dw^k_m(\vqq)$.
\item \label{item-veloctransf} 
$\vvh_m=\vvh_k \circ dw^k_m$ and $\dom\vvh_m=\dom\vvh_k$.
\item \label{item-acctransf} 
$\vah_m=\vah_k\circ dw^k_m$ and $\dom\vah_m=\dom\vah_k$.
\end{enumerate}
\end{prop}

\begin{proof}
To prove Item \eqref{item-mort}, observe that $\mu\big(dw^k_m(\vpp)\big)=\mu(\vpp)$ by Theorem \ref{thm-poi}.
The statement $\vpp\mort\vqq$ iff $\mu(\vpp+\vqq)^2=\mu(\vpp)^2+\mu(\vqq)^2$ can be shown by straightforward calculation.
Thus Item \eqref{item-mort} is clear since $dw^k_m$ is linear by Theorem \ref{thm-poi}.

To prove Items \eqref{item-veloctransf} and \eqref{item-acctransf}, let us note that $lc^h_m$ and $lc^h_k$ are functions by Item \eqref{item-trfunct} in Proposition \ref{prop-tr}.
Thus $\vvh_m=\vvh_k \circ dw^k_m$ follows by chain rule because $lc^h_k= lc^h_m\circ w^m_k$ (by \eqref{item-tr} in Proposition \ref{prop-tr}), the derivative of $w^k_m$ is $dw^k_m$ (since $w^k_m$ is affine transformation by Theorem \ref{thm-poi}), and $\vvh_x=(lc^h_x)'$ (by definition).
Hence $\dom\vvh_m=\dom\vvh_k$ also holds since $dw^k_m$ is a bijection.
$\vah_m=\vah_k\circ dw^k_m$ follows from \eqref{item-veloctransf} of this Proposition by chain rule because the derivative of $dw^k_m$ is $dw^k_m$ (since $dw^k_m$ is linear transformation), and $\vah_x=(\vvh_x)'$ (by definition). Hence $\dom\vah_m=\dom\vah_k$ also holds since $dw^k_m$ is a bijection.
\end{proof}

\begin{prop}$\,$
\label{prop-vmorta}
\begin{enumerate}
\item Let $\alpha$ be a well-parametrized timelike curve.
If $\alpha$ is twice differentiable at $t\in\dom\alpha$, then $\alpha'(t)\mort\alpha''(t)$.
\item Let $d\ge 3$.
Assume \ax{AccRel^0_\d}.
Let $k\in\Ob$ and $m\in \IOb$.
Then $\vvkm(t)\mort\vakm(t)$ for all $t\in \dom \vakm$.
\end{enumerate}
\end{prop}

\begin{proof}
To prove Item (1), let $t\in\dom\alpha$ such that $\alpha$ is twice differentiable at $t$.
Since $\alpha$ is a well-parametrized timelike curve, we have that
\begin{equation}\label{eq-wellp}
\big(\alpha'_1(t)\big)^2-\big(\alpha'_2(t)\big)^2-\ldots-\big(\alpha'_d(t)\big)^2=1.
\end{equation}
By derivation of both sides of equation \eqref{eq-wellp} we have that
\begin{equation}\label{eq-mort}
2\alpha'_1(t)\cdot\alpha''_1(t)-2\alpha'_2(t)\cdot\alpha''_2(t)-\ldots-2\alpha'_d(t)\cdot\alpha''_d(t)=0.
\end{equation}
Thus $\alpha'(t)\mort\alpha''(t)$, which is what we wanted to prove.

Item (2) is an easy consequence of Item (1) since $\vvkm=(lc^k_m)'$, $\vakm=(lc^k_m)''$, $lc^k_m$ is a well-parametrized timelike curve by Proposition \ref{prop-wellpar}, and $lc^k_m$ is twice differentiable at $t$ iff $t\in \dom \vakm$.
\end{proof}

It is practical to introduce a notation for the next vertical plane 
\begin{equation}
\Df{Plane(t,x)}\leteq\setopen \vpp\in \Q^d\::\: p_3=\ldots=p_d=0\setclose.
\end{equation}
If $f:\Q\parrow\Q$, we abbreviate $f(t)>0$ for all $t\in\dom f$ to $f>0$.
We also use the analogous notation $f<0$.

\begin{lem}$\,$ 
\label{lem-vmon}
\begin{enumerate}
\item Assume \ax{IND}.
Let $\alpha$ be a definable and twice differentiable timelike curve such that $\ran\alpha\subset Plane(t,x)$.
If $\alpha''\circ\mu<0$, then $\alpha'_2$ is increasing or decreasing.
\item Let $d\ge 3$.
Assume \ax{AccRel_\d}.
Let $k\in\Ob$ and $m\in \IOb$ such that $wl_m(k)\subset Plane(t,x)$ and $\dom\vakm=\dom lc^k_m$.
If $k$ is positively accelerated, $(\vvkm)_2$ is increasing or decreasing.
\end{enumerate}
\end{lem}
\begin{proof}
To prove Item (1), let $t\in\dom \alpha$.
By Proposition \ref{prop-vmorta}, $\alpha''(t)$ is a spacelike vector since it is Minkowski orthogonal to a timelike one.
Therefore, $\mu(\alpha''(t))<0$ iff $|\alpha''_\sigma(t)|\neq0$.
Thus since $\ran\alpha\subset Plane(t,x)$, we have that $\mu(\alpha''(t))<0$ iff $\alpha''_2(t)\neq0$.
Thus by \ax{IND}-Darboux's Theorem, $\alpha''\circ\mu<0$ iff $\alpha''_2>0$ or $\alpha''_2<0$ since $\alpha'_2$ is definable and $\dom\alpha'=\dom\alpha$ is connected.
Then by \ax{IND}-Mean-Value Theorem, $\alpha'_2$ is increasing or decreasing.

Item (2) is an easy consequence of Item (1) because of the following.
Let $\alpha=lc^k_m$.
Then by Proposition \ref{prop-wellpar}, $\alpha$ is definable (well-parametrized) timelike curve.
$\alpha$ twice differentiable since $\dom\alpha''=\dom\vakm=\dom lc^k_m=\dom\alpha$;
and $\ran\alpha\subset Plane(t,x)$ since by \eqref{item-rantr} in Proposition \ref{prop-tr}, $wl_m(k)=\ran lc^k_m$.
Then $k$ is positively accelerated iff $\alpha''\circ\mu<0$.
Hence by Item (1) if $k$ is positively accelerated, $(\vvkm)_2=\alpha'_2$ is increasing or decreasing.
\end{proof}

The \df{light cone} of $\vpp\in\Q^d$ is defined as 
$\Df{\Lambda}{}_{\vpp}\leteq\setopen \vqq\in\Q^d:\vpp\pheq\vqq \setclose$.
The \df{past light cone} of $\vpp\in\Q^d$ is defined as 
$\Df{\Lambda}{}^-_{\vpp}\leteq\setopen\vqq\in\Q^d\::\: \vpp\pheq\vqq \and q_\tau\le p_\tau\setclose$.
The \df{future light cone} of $\vpp\in\Q^d$ is defined as 
$\Df{\Lambda}{}^+_{\vpp}\leteq\setopen\vqq\in\Q^d\::\: \vpp\pheq\vqq \and q_\tau\ge p_\tau\setclose$.
We say that $\vpp\in\Q^d$ \df{chronologically precedes} $\vqq\in\Q^d$, in symbols $\vpp\Df{\ll}\vqq$, iff $\vpp\teq\vqq$ and $p_\tau<q_\tau$.
The \df{chronological past} of $\vpp\in\Q^d$ is defined as $\Df{I}{}^-_{\vpp}\leteq\Setopen\vqq\in\Q^d\::\: \vqq\ll\vpp\Setclose$.
The \df{chronological future} of $\vpp\in\Q^d$ is defined as $\Df{I}{}^+_{\vpp}\leteq\Setopen\vqq\in\Q^d\::\: \vpp\ll\vqq \Setclose$.
The \df{chronological interval} between $\vpp\in\Q^d$ and $\vqq\in\Q^d$ is defined as $\Df{\llangle\vpp, \vqq\rrangle}\leteq \setopen \vr\in\Q^d\::\: \vpp\teq\vr \and \vqq\teq\vr \and r_\tau\in (p_\tau,q_\tau)\setclose$.
We also use the notation $\Df{I}{}_{\vpp}\leteq I^-_{\vpp}\cup I^+_{\vpp} \cup \{\vpp\}$.

\begin{lem}
\label{lem-cau}
Let $\vpp,\vqq\in\Q^d$.
Then
\begin{enumerate}
\item \label{item-distjcones} If $\vpp\teq \vqq$, then $\Lambda^-_{\vpp}\cap \Lambda^-_{\vqq}=\Lambda^+_{\vpp}\cap \Lambda^+_{\vqq}=\emptyset$.
\item If $\vpp\ll\vqq$, then $\Lambda^-_{\vqq}\cap I^-_{\vpp}=\emptyset$, and $\Lambda^-_{\vpp}\cup I^-_{\vpp}\subset I^-_{\vqq}$.
\item $\vpp\ll\vqq$ iff $I^+_{\vpp}\cap I^-_{\vqq}\neq\emptyset$.\qed
\end{enumerate}
\end{lem}

\begin{lem}
\label{lem-chord}
Assume \ax{IND}.
Let $\gamma$ be a definable timelike curve, and let $x,y\in \dom\gamma$ such that $x\neq y$.
Then
\begin{enumerate}
\item\label{item-chord} All the chords of $\gamma$ are timelike, that is, $\gamma(x)\teq\gamma(y)$.
\item\label{item-pcone} If $\gamma(x)\in I^-_{\vpp}$ and $\gamma(y)\not\in I^-_{\vpp}$, there is a $z\in[x,y]$ such that $\gamma(z)\in\Lambda^-_{\vpp}$.
\item\label{item-fcone} If $\gamma(x)\in I^+_{\vpp}$ and $\gamma(y)\not\in I^+_{\vpp}$, there is a $z\in[x,y]$ such that $\gamma(z)\in\Lambda^+_{\vpp}$.
\item\label{item-cord} If $\gamma_\tau$ is increasing (decreasing), $\gamma(x)\ll\gamma(y)$ iff $x<y$ ($y<x$).
\item\label{item-cinv} $z\in(x,y)$ iff $\gamma(z)\in \llangle \gamma(x),\gamma(y)\rrangle$.
\end{enumerate}
\end{lem}
\begin{proof}
For the proof of Item \eqref{item-chord}, see Proposition A.0.15 in \cite{mythes}.

To prove Item \eqref{item-pcone}, let 
\begin{equation}
H\leteq\Setopen t\in[x,y]\::\: \mu\big(\gamma(t),\vpp\big)<0 \and \gamma(t)_\tau<p_\tau\Setclose.
\end{equation}
It is clear that $H\subseteq \dom \gamma$ is definable, bounded and nonempty.
Let $z\leteq\sup H$ which exists by \ax{IND}.
Thus by continuity of $t\mapsto \mu(\gamma(t),\vpp)$ and $\gamma_\tau$, we have that $\gamma(z)\not\in I^-_{\vpp}$ since $z$ is an upper bound of $H$; furthermore $\mu(\gamma(t),\vpp)\le0$ and $\gamma(t)_\tau\le p_\tau$ since $z$ is the least upper bound of $H$.
But $\gamma(t)_\tau=p_\tau$ and $\mu(\gamma(t),\vpp)<0$ is impossible.
Thus $\gamma(t)_\tau\le p_\tau$ and $\mu(\gamma(t),\vpp)=0$.
Hence $\gamma(\vpp)\in\Lambda^-_{\vpp}$.

Item \eqref{item-fcone} is clear from Item \eqref{item-pcone} since the continuous bijection $\vpp\mapsto -\vpp$ takes $I^+_{\vpp}$ to $I^-_{\vpp}$ and $\Lambda^+_{\vpp}$ to $\Lambda^-_{\vpp}$.

Item \eqref{item-cord} is clear by Item \eqref{item-chord}.

Item \eqref{item-cinv} is an easy consequence of Item \eqref{item-cord} since $\gamma_\tau$ is either increasing or decreasing by Lemmas \ref{lem-inj} and \ref{lem-nice}.
\end{proof}

We use the following notations: 
\begin{equation}
\Df{B_\varepsilon(\vpp)}\leteq\Setopen \vqq\in\Q^n\::\: |\vpp-\vqq|<\varepsilon\Setclose, 
\end{equation}
\begin{equation}
\Df{line(\vpp,\vqq)}\leteq\Setopen \vpp+\lambda(\vpp-\vqq) \::\:\lambda\in\Q\Setclose,
\end{equation}
\begin{equation}
\Df{Cone_{\varepsilon}(\vpp;\vqq)}\leteq\bigcup_{\vr\,\in B_\varepsilon(\vqq)}line(\vpp,\vr\,) \text{ and}
\end{equation}
\begin{equation}
\Df{\Lambda^-[H]}\leteq\bigcup_{\vpp\in H}\Lambda^-_{\vpp}.
\end{equation}

Let $\alpha$ and $\beta$ be timelike curves.
We say that $\beta_*$ is the \df{photon reparametrization of $\beta$ according to $\alpha$} if 
\begin{equation}
\beta_*=\setopen \langle t,\vpp\rangle\in\dom\alpha\times\ran\beta \::\:\vpp\in\Lambda^-_{\alpha(t)}\setclose.
\end{equation}

\begin{prop}
\label{prop-ph}
Assume \ax{IND}.
Let $\alpha$ and $\beta$ be definable timelike curves.
Let $\beta_*$ be the photon reparametrization of $\beta$ according to $\alpha$.
\begin{enumerate} 
\item \label{item-phcont} Then $\beta_*$ is a definable, continuous and injective curve.
\item \label{item-phder} If $\ran\alpha\cap\ran\beta=\emptyset$, and $\ran \alpha\cup \ran \beta$ is in a vertical plane, $\beta_*$ is a timelike curve, and $\beta_*(t_0)+\beta'_*(t_0)\in\Lambda^-_{\alpha(t_0)+\alpha'(t_0)}$.
\end{enumerate}
\end{prop}

\begin{figure}[h!btp]
\small
\begin{center}
\psfrag{p}[l][l]{$\vpp$}
\psfrag{q}[bl][bl]{$\vqq$}
\psfrag{r}[l][l]{$\vr$}
\psfrag{ph}[l][l]{$ph$}
\psfrag{e}[tr][tr]{$\varepsilon$}
\psfrag{e1}[tr][tr]{$\varepsilon_1$}
\psfrag{e2}[bl][bl]{$\varepsilon_2$}
\psfrag{a}[l][l]{$\alpha$}
\psfrag{a'0}[bl][bl]{$\alpha'(t_0)$}
\psfrag{b}[tr][tr]{$\beta,\beta_*$}
\psfrag{b'0}[br][br]{$\beta'(\bar{t}_0)$}
\psfrag{a0}[l][l]{$\alpha(t_0)$}
\psfrag{b0}[l][l]{$\beta(\bar{t}_0)=\beta_*(t_0)$}
\psfrag{adif}[bl][bl]{$\frac{\alpha(t)-\alpha(t_0)}{t-t_0}$}
\psfrag{b*'0}[br][br]{$\beta_*'(t_0)$}
\psfrag{b*t}[tl][tl]{$\beta_*(t)$}
\psfrag{at}[l][l]{$\alpha(t)$}
\includegraphics[keepaspectratio, width=\textwidth]{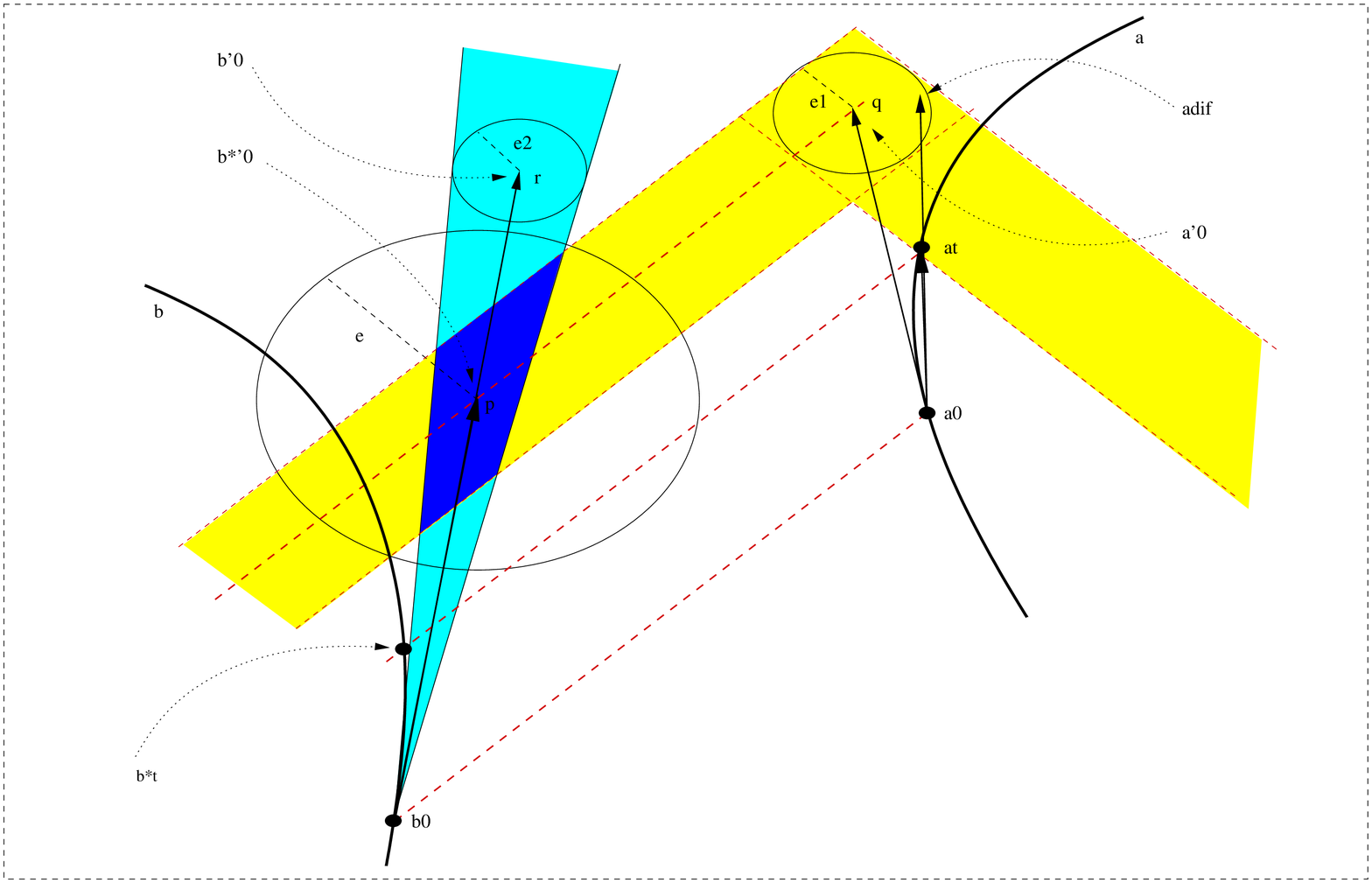}
\caption{\label{figphpar} Illustration for the proof of Proposition \ref{prop-ph}.}
\end{center}
\end{figure}

\begin{proof}
It is clear that $\beta_*$ is definable.

To show that $\beta_*$ is a function, we need to prove that $\Lambda^-_{\alpha(t)}\cap\ran \beta$ has one element at the most for all $t\in\dom \alpha$.
It is clear by Lemma \ref{lem-chord} since if it had two distinct elements, say $\vpp$ and $\vqq$, then $\setopen \vpp,\vqq\setclose$ would not be a timelike chord of $\beta$, but a lightlike one.

For all $t\in\dom\beta_*$, let $\bar{t}\in\dom\beta$ such that $\beta(\bar t\,)=\beta_*(t)$, and let $f:t \mapsto \bar t$ be the reparametrization map, that is, $f\leteq\beta_*\circ\beta^{-1}$.
First we show that 
\begin{multline*} t\in (t_1,t_2) \iff
\alpha(t)\in\llangle \alpha(t_1),\alpha(t_2)\rrangle \stackrel{(*)}{\iff} \\
\beta_*(t)\in\llangle \beta_*(t_1),\beta_*(t_2)\rrangle
\iff \bar{t}\in (\bar{t}_1,\bar{t}_2)
\end{multline*}
if $t,t_1,t_2\in\dom\beta_*$.
The first and the last equivalence are clear by \eqref{item-cinv} in Lemma \ref{lem-chord} since $\alpha$, 
$\beta$ are timelike curves and $\beta(\bar{t}\,)=\beta_*(t)$ for all $t\in\dom\beta_*$.
To prove $(*)$, we can assume that $\alpha(t_1)\ll\alpha(t)\ll\alpha(t_2)$.
Thus $\beta_*(t)\ll\alpha(t_2)$ since $\Lambda^-_{\alpha(t)}\subset I^-_{\alpha(t_2)}$ (2) of by Lemma \ref{lem-cau}.
Therefore, $\beta_*(t)\ll\beta_*(t_2)$ since $\beta_*(t)\in I_{\beta_*(t_2)}$ by \eqref{item-chord} in Lemma \ref{lem-chord}, but $I^+_{\beta_*(t_2)}\cap I^-_{\alpha(t_2)}=\emptyset$ (3) by Lemma \ref{lem-cau}.
A similar argument can show that $\beta_*(t_1)\ll\beta_*(t)$, so $(*)$ is proved.
Now we have that $f$ preserves betweenness, so it is monotonous.

To show that $\dom \beta_*$ is connected, let $x,y\in\dom\beta_*$, and let $z\in(x,y)$.
Then $z\in\dom\alpha$ since $x,y\in\dom \alpha$ and $\dom\alpha$ is connected.
Since $\alpha$ is a timelike curve, $\alpha(z)\in\llangle \alpha(x),\alpha(y)\rrangle$.
Without losing generality, we can assume that $\alpha(x)\ll \alpha(z)\ll\alpha(y)$.
Then $\beta_*(x)\in I^-_{\alpha(z)}$ since $\beta_*(x)\in\Lambda^-_{\alpha(x)}\subset I^-_{\alpha(z)}$;
and $\beta_*(y)\not\in I^-_{\alpha(z)}$ since $\beta_*(y)\in\Lambda^-_{\alpha(y)}$ and $\Lambda^-_{\alpha(y)}\cap I^-_{\alpha(z)}=\emptyset$, see Lemma \ref{lem-cau}.
Then by \eqref{item-pcone} in Lemma \ref{lem-chord}, there is a $\hat{z}\in\dom\beta$ such that $\beta(\hat{z})\in\Lambda^-_{\alpha(z)}$ since $\beta(\bar{x})\in I^-_{\alpha(z)}$ and $\beta({\bar{y}})\not\in I^-_{\alpha(z)}$.
Thus $\langle z,\beta(\hat{z})\rangle\in\beta_*$.
Consequently, $z\in\dom\beta_*$.
Hence $\dom\beta_*$ is connected.

Now using a similar argument, we show that $\ran f\subseteq\dom\beta$ is also connected.
To do so, let $\bar{x},\bar{y}\in\ran f$ and $\hat{z}\in(\bar{x},\bar{y})$.
Then $\hat{z}\in\dom \beta$.
We can assume that $\beta(\bar{x})\ll\beta(\hat{z})\ll\beta((\bar{y})$.
Then $\alpha(x)\in I^+_{\beta(\hat{z})}$ and $\alpha(y)\not\in I^+_{\beta(\hat{z})}$.
Thus there is a $z\in\dom\alpha$ such that $\alpha(z)\in\Lambda^+_{\beta(\hat{z})}$.
Consequently, $\beta(\hat{z})\in\Lambda^-_{\alpha(z)}$, so $\langle z,\beta(\hat{z})\rangle\in\beta_*$.
Therefore, $\hat{z}\in\ran f$, and hence $\ran f$ is connected.

Since $\ran f$ is connected and $f$ is monotonous, $f$ must be continuous by Lemma \ref{lem-moncont}.
Hence $\beta_*=f\circ\beta$ is also continuous and $\beta_*$ injective since both $\beta$ and $f$ are such.
So Item \eqref{item-phcont} is proved.

To prove Item \eqref{item-phder}, let $\vqq=\alpha'(t_0)+\alpha(t_0)$, $\vr=\beta'(\bar{t}_0)+\beta(\bar{t}_0)$, and let $\vpp$ be the unique element of $\Lambda^-_{\vqq}\cap line(\beta(\bar{t}_0),\vr\,)$, see Figure \ref{figphpar}.
We will show that $\beta'_*(t_0)=\vpp-\beta_*(t_0)$.
To do so, let $\varepsilon\in\Q^+$ be fixed.
We have to show that there is a $\delta\in\Q^+$ such that $\frac{\beta_*(t)-\beta_*(t_0)}{t-t_0}\in B_\varepsilon(\vpp)$ if $t\in\dom\beta_*\cap B_\delta(t_0)$.
It is clear that we can choose $\varepsilon_1$ and $\varepsilon_2$ such that 
\begin{equation}\label{eq-inball}
\Lambda^-[B_{\varepsilon_1}(\vqq)]\cap Cone_{\varepsilon_2}(\beta_*(t_0);\vr\,)\subset B_\varepsilon(\vpp).
\end{equation}
Since $\alpha$ is differentiable at $t_0$, there is a $\delta_1\in \Q^+$ such that
\begin{equation}\label{eq-alpha}
\frac{\alpha(t)-\alpha(t_0)}{t-t_0}+\alpha(t_0)\in B_{\varepsilon_1}(\vqq)
\end{equation}
if $t\in\dom\alpha\cap B_{\delta_1}(t_0)$.
Since $\ran\beta\cap\ran\alpha=\emptyset$, and $\ran\beta\cup\ran\alpha$ is in a vertical plane, $line\big(\beta_*(t),\alpha(t)\big)$ and $line\big(\beta_*(t_0),\alpha(t_0)\big)$ are parallel.
Hence
\begin{equation}
\frac{\beta_*(t)-\beta_*(t_0)}{t-t_0}+\beta_*(t_0)\in \Lambda^-_{\frac{\alpha(t)-\alpha(t_0)}{t-t_0}+\alpha(t_0)}.
\end{equation}
Thus by \eqref{eq-alpha}, we have that 
\begin{equation}\label{eq-inpast}
\frac{\beta_*(t)-\beta_*(t_0)}{t-t_0}+\beta_*(t_0)\in \Lambda^-[B_{\varepsilon_1}(\vqq)]
\end{equation}
if $t\in\dom\beta_*\cap B_{\delta_1}(t_0)$.
Since $\beta$ is differentiable at $\bar{t}_0$, there is a $\bar{\delta}_2\in \Q^+$ such that 
\begin{equation} \label{eq-beta}
\frac{\beta(\bar{t}\,)-\beta(\bar{t}_0)}{\bar{t}-\bar{t}_0}+\beta(\bar{t}_0)\in B_{\varepsilon_2}(\vr\,)
\end{equation}
if $\bar{t}\in\dom\beta\cap B_{\bar{\delta}_2}(\bar{t}_0)$.
Since $f:t\mapsto\bar{t}$ is continuous, there is a $\delta_2\in\Q^+$ such that \eqref{eq-beta} holds if $t\in\dom\beta_*\cap B_{\delta_2}(t_0)$.
Since
\begin{equation}
\frac{\beta_*(t)-\beta_*(t_0)}{t-t_0}=\frac{\beta(\bar{t}\,)-\beta(\bar{t}_0)}{\bar{t}-\bar{t}_0}\cdot\frac{\bar{t}-\bar{t}_0}{t-t_0},
\end{equation}
we have that 
\begin{equation}\label{eq-incone}
\frac{\beta_*(t)-\beta_*(t_0)}{t-t_0}+\beta_*(t_0)\in Cone_{\varepsilon_2}(\beta_*(t_0);\vr\,)
\end{equation}
if $t\in\dom\beta_*\cap B_{\delta_2}(t_0)$.
Let $\delta=\min(\delta_1,\delta_2)$.
Therefore, by equations \eqref{eq-inpast} and \eqref{eq-incone}, we have that 
\begin{equation}
\frac{\beta_*(t)-\beta_*(t_0)}{t-t_0}+\beta_*(t_0)\in \Lambda^-[B_{\varepsilon_1}(\vqq)]\cap Cone_{\varepsilon_2}(\beta_*(t_0);\vr\,)
\end{equation}
if $t\in\dom\beta_*\cap B_{\delta}(t_0)$.
But the latter is a subset of $B_\varepsilon(\vpp)$ by equation \eqref{eq-inball}.
Consequently,
\begin{equation}
\frac{\beta_*(t)-\beta_*(t_0)}{t-t_0}+\beta_*(t_0)\in B_\varepsilon(\vpp)
\end{equation}
if $t\in\dom\beta_*\cap B_{\delta}(t_0)$.
Hence $\beta_*$ is differentiable at $t_0$ and $\beta'_*(t_0)=\vp-\beta_*(t_0)$, as desired.
\end{proof}

Let $\vpp,\vqq\in Plane(t,x)$.
Then the \df{photon sum} of $\vpp$ and $\vqq$, in symbols $\vpp\Df{\phsum}\vqq$, is the intersection of the two photon lines $\setopen \vpp+\langle A,A,0,\ldots 0\rangle:A\in\Q \setclose$ and $\setopen \vqq+\langle B,-B,0,\ldots 0\rangle:B \in\Q \setclose$.

\begin{figure}[h!btp]
\small
\begin{center}
\psfrag{ph1}[l][l]{$ph_1$}
\psfrag{ph2}[r][r]{$ph_2$}
\psfrag{p2}[r][r]{$\left\langle \frac{p_\tau-p_2}{2}, -\frac{p_\tau-p_2}{2}, \vo\,\right\rangle$}
\psfrag{p}[br][br]{$\langle p_\tau, p_2,\vo\,\rangle=\vpp$}
\psfrag{q1}[l][l]{$\left\langle\frac{q_\tau+q_2}{2}, \frac{q_\tau+q_2}{2}, \vo\,\right\rangle$}
\psfrag{q}[bl][bl]{$\vqq=\langle q_\tau, q_2,\vo\,\rangle$}
\psfrag{qlp}[l][l]{$\vqq\phsum\vpp$}
\psfrag{plq}[l][l]{$\vpp\phsum\vqq$}
\psfrag{t}[l][l]{$t$}
\psfrag{x}[l][l]{$x$}
\psfrag{pt}[bl][bl]{$p_\tau$}
\psfrag{ps}[t][t]{$p_\sigma$}
\psfrag{qt}[tr][tr]{$q_\tau$}
\psfrag{qs}[t][t]{$q_\sigma$}
\includegraphics[keepaspectratio, width=\textwidth]{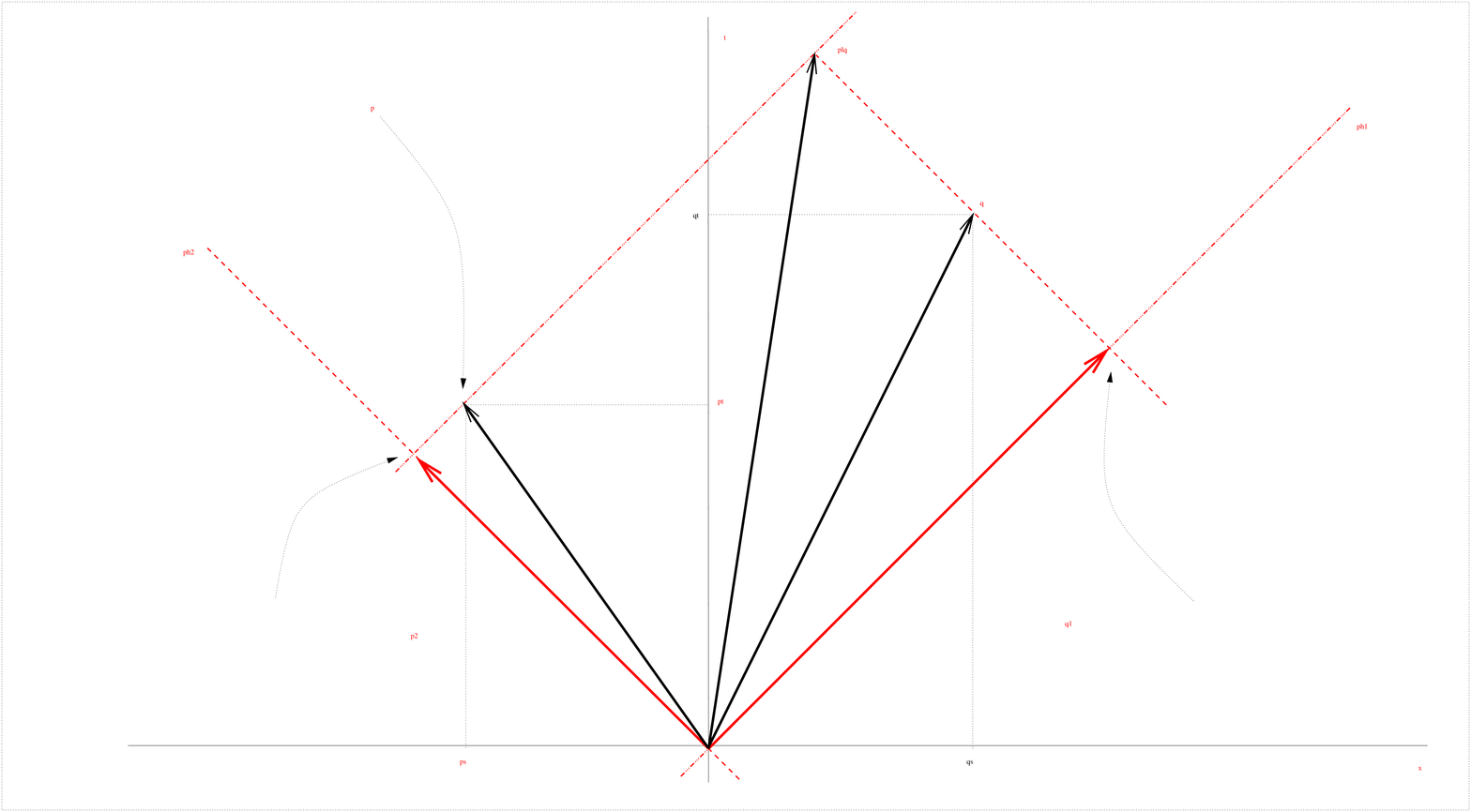}
\caption{\label{figphcrd} Illustration of the photon sum $\vpp\phsum\vqq$, and for the proof of Lemma \ref{lem-phsum}.}
\end{center}
\end{figure}

\begin{lem}
\label{lem-phsum} 
Let $\vpp,\vqq\in Plane(t,x)$, and let $a=\frac{q_\tau+q_2}{2}$ and $b=\frac{p_\tau-p_2}{2}$.
Then $\vpp\phsum\vqq=\langle a+b,a-b,0,\ldots,0\rangle$.
\end{lem}
\begin{proof}
The proof is straightforward, see Figure \ref{figphcrd}.
\end{proof}

\begin{lem}
\label{lem-tlinvcont}
Assume \ax{IND}.
Let $\beta$ be a definable timelike curve.
Then $\beta^{-1}:\ran\beta\rightarrow \dom\beta$ is definable, injective and continuous.
\end{lem}

\begin{proof}
It is clear that $\beta^{-1}$ is definable and injective.

Since by Lemmas \ref{lem-tlnice} and \ref{lem-inj}, $\beta$ is injective, $\beta^{-1}$ is a function from $\ran\beta$ to $\dom\beta$.
To prove that it is also continuous, let $t_0\in\dom\beta$.
We have to show that for all $\varepsilon \in \Q^+$, there is a $\delta\in\Q^+$ such that if $t\in \dom\beta$ and $|\beta(t)-\beta(t_0)|<\delta$, then $|t-t_0|<\varepsilon$.
By Lemma \ref{lem-chord}, $t\in(t_0-\varepsilon,t_0+\varepsilon)$ iff $\beta(t)\in\llangle \beta(t_0-\varepsilon),\beta(t_0+\varepsilon)\rrangle$.
Thus since $\llangle \beta(t_0-\varepsilon),\beta(t_0+\varepsilon)\rrangle$ is an open set, there is a good $\delta$.
\end{proof}

\begin{lem}
\label{lem-repar}
Assume \ax{IND}.
Let $\beta$ be a definable timelike curve and $\beta_*$ a definable continuous curve such that $\ran \beta_*\subseteq \ran\beta$, and let $f\leteq\beta_*\circ\beta^{-1}$.

\begin{enumerate}
\item Then $f$ is a definable and continuous function.
\item If $\beta_*$ is injective, $f$ is also injective.
Moreover, $\dom f$ and $\ran f$ are connected and $f^{-1}$ is also a definable, monotonous and continuous function.
\item If $\beta_*$ is differentiable such that $\beta_*'(t)\neq \vo$ for all $t\in\dom\beta$, then $f$ is injective and differentiable, and $f'(t)\neq0$.
Hence $f^{-1}$ is also a differentiable function.
\end{enumerate}
\end{lem}

\begin{proof}
Item (1) is clear by Lemma \ref{lem-tlinvcont}.

Item (2) is clear by Item (1) and Lemma \ref{lem-injcont} since $\dom f=\dom \beta_*$ which is connected.

To prove Item (3), let $t_0\in\dom f$.
Since $\ran\beta_*\subseteq\ran\beta$, we have that there is a $\lambda\in\Q$ such that $\lambda\cdot\beta'(t_0)=\beta'_*(f(t_0))$.
Since $\big(f(t)-f(t_0)\big)/(t-t_0)$ is the ratio of parallel vectors 
\begin{equation}
\frac{\beta(t)-\beta(t_0)}{t-t_0} \quad\text{ and }\quad \frac{\beta_*\big(f(t)\big)-\beta_*\big(f(t_0)\big)}{f(t)-f(t_0)},
\end{equation} 
we have that $\big(f(t)-f(t_0)\big)/(t-t_0)$ tends to $\beta'(t_0)/\beta'_*(f(t_0))=1/\lambda$ if $t$ tends to $t_0$.
Thus $f$ is differentiable, and $f'(t_0)=1/\lambda$.
\end{proof}

\begin{lem}
\label{lem-cordmap}
Assume \ax{IND}.
Let $\alpha$ be a definable timelike curve.
Let $t\in\dom\alpha$ and $x=\alpha_\tau(t)$.
Let $f_\alpha\leteq\alpha^{-1}_\tau\circ\alpha_\sigma$.
\begin{enumerate}
\item Then $f_\alpha$ is a differentiable curve, and $f'_\alpha(x)=\alpha'_\sigma(t)/\alpha'_\tau(t)$.
\item If $\alpha$ is twice differentiable at $t$, then so is $f_\alpha$ at $x$, and 
\begin{equation}f''_\alpha(x)=\frac{\alpha'_\tau(t)\alpha''_\sigma(t)-\alpha''_\tau(t)\alpha'_\sigma(t)}{\alpha'_\tau(t)^3}.\end{equation}
\end{enumerate}
\end{lem}

\begin{proof}
Let us first prove Item (1).
We have that $\alpha_\tau$ is injective by Lemmas \ref{lem-tlnice} and \ref{lem-inj}.
Hence $f_\alpha$ is a function.
$\dom f_\alpha$ is connected since $\dom f_\alpha=\ran\alpha_\tau$ and $\ran\alpha_\tau$ is connected by Lemma \ref{lem-injcont}.
Thus $f_\alpha$ is a curve.
Since $\alpha_\tau$ is an injective differentiable curve, $\alpha^{-1}_\tau$ is also such and $(\alpha^{-1}_\tau)'(x)=1/\alpha'_\tau(t)$.
Thus by chain rule, we have that $f'_\alpha(x)=\alpha'_\sigma(t)/\alpha'_\tau(t)$.

Now let us prove Item (2).
If $\alpha$ is twice differentiable at $t$, then so are $\alpha_\sigma$ and $\alpha_\tau$.
By Item (1), $f'_\alpha=\alpha^{-1}_\tau\circ\alpha'_\sigma/\alpha'_\tau$.
Thus $f_\alpha$ is twice differentiable at $x$ and a straightforward calculation based on the rules of differential calculus can show that $f''_\alpha(x)$ is what was stated.
\end{proof}

\begin{lem}
\label{lem-twocurve}
Assume \ax{IND}.
Let $\alpha$ and $\beta$ be definable timelike curves such that $\ran\alpha\cup\ran\beta$ is in a vertical plane.
Let $t_1,t_2\in\dom\alpha$ and $\bar{t}_1,\bar{t}_2\in\dom\beta$ such that $\alpha(t_1)\seq\beta(\bar{t}_1)$, $\alpha(t_2)\seq\beta(\bar{t}_2)$ and $\big(\beta(\bar{t}_1)-\alpha(t_1)\big)\upp\big(\alpha(t_2)-\beta(\bar{t}_2)\big)$.
Then there is a $t\in(t_1,t_2)$ such that $\alpha(t)\in\ran\beta$.
Hence $Ran\alpha\cap\ran\beta\neq\emptyset$.
\end{lem}

\begin{proof}
Since $\ran\alpha\cup\ran\beta$ is in a vertical plane, we can assume, without losing generality, that $d=2$.
By \ax{IND}-Bolzano's theorem, we can also assume that $\alpha(t_1)_\tau=\beta(\bar{t}_1)_\tau$ and $\alpha(t_2)_\tau=\beta(\bar{t}_2)_\tau$.
Let $x_1=\alpha(t_1)_\tau$ and $x_2=\alpha(t_2)_\tau$.
Let $f_\alpha\leteq \alpha_\tau^{-1}\circ\alpha_\sigma$ and $f_\beta\leteq\beta_\tau^{-1}\circ\beta_\sigma$.
Then $f_\alpha$ and $f_\beta$ are continuous curves, see Lemma \ref{lem-cordmap}.
By the assumption $\big(\beta(\bar{t}_1)-\alpha(t_1)\big)\upp\big(\alpha(t_2)-\beta(\bar{t}_2)\big)$, we have that $\big(f_\beta(x_1)-f_\alpha(x_1)\big)\big(f_\alpha(x_2)-f_\beta(x_2)\big)<0$.
Thus by \ax{IND}-Bolzano's theorem, there is an $x\in(x_1,x_2)$ such that $f_\alpha(x)=f_\beta(x)$.
Let $t\leteq \alpha_\tau^{-1}(x)$.
Then $\alpha(t)\in\ran\beta$.
\end{proof}

Let $\alpha$ and $\beta$ be timelike curves.
We say that $\beta_*$ is the \df{radar reparametrization of $\beta$ according to $\alpha$} if 
\begin{equation}
\beta_*=\setopen \langle t,\vpp\rangle\in\dom\alpha\times\ran\beta \::\:\exists r\in\Q \quad \vpp\in\Lambda^-_{\alpha(t+r)}\cap\Lambda^+_{\alpha(t-r)}\setclose.
\end{equation}
We say that $\beta$ is constant $r$ radar distance from $\alpha$ iff 
\begin{equation}
\ran\beta\subseteq \bigcup_{t\pm r\in \dom\alpha}\Lambda^-_{\alpha(t+r)}\cap\Lambda^+_{\alpha(t-r)}.
\end{equation}
We note that this $r$ can be negative if $\alpha_\tau$ is decreasing since by this definition, $\alpha(t-r)\ll\alpha(t+r)$.

We will also use the notation
\begin{equation}
\Df{\vex}=\langle 0,1,0,\ldots,0\rangle.
\end{equation}

\begin{prop}
\label{prop-rad}
Assume \ax{IND}.
Let $\alpha$ and $\beta$ be definable timelike curves.
Let $\beta_*$ be the radar reparametrization of $\beta$ according to $\alpha$.
\begin{enumerate} 
\item\label{item-radcont} Then $\beta_*$ is a definable, injective, and continuous curve.
\item\label{item-raddiff} If $\ran \alpha\cup \ran \beta$ is in a vertical plane, and $\beta$ is constant $r$ radar distance from $\alpha$, then $\beta_*$ is differentiable.
\item\label{item-radder} Let us further assume that this vertical plane is $Plane(t,x)$.
Then 
\begin{itemize}
\item[] $\beta'_*(t)=\alpha'(t-r)\phsum\alpha'(t+r)\enskip\text{ iff }\enskip \big(\beta_*(t)-\alpha(t)\big)\upp\phantom{-}\vex$, 
\item[] $\beta'_*(t)=\alpha'(t+r)\phsum\alpha'(t-r)\enskip\text{ iff }\enskip \big(\beta_*(t)-\alpha(t)\big)\upp -\vex$.
\end{itemize}
\end{enumerate}
\end{prop}

\begin{proof}
It is clear that $\beta_*$ is definable.

\begin{figure}[h!btp]
\small
\begin{center}
\psfrag{df}[r][r]{$\dom f$}
\psfrag{t}[l][l]{$t$}
\psfrag{dg}[r][r]{$\dom g$}
\psfrag{rh}[r][r]{$\ran h$}
\psfrag{g-1tt}[l][l]{$g^{-1}(\tilde{t})$}
\psfrag{htt}[l][l]{$h(\tilde{t})$}
\psfrag{f}[b][b]{$f$}
\psfrag{g}[bl][bl]{$g$}
\psfrag{h}[br][br]{$h$}
\psfrag{a}[tl][tl]{$\alpha$}
\psfrag{b}[tl][tl]{$\beta,\beta_*$}
\psfrag{tt}[l][l]{$\tilde{t}$}
\psfrag{dh}[l][l]{$\dom h$}
\psfrag{rf}[r][r]{$\ran f$}
\psfrag{rg}[l][l]{$\ran g$}
\includegraphics[keepaspectratio, width=\textwidth]{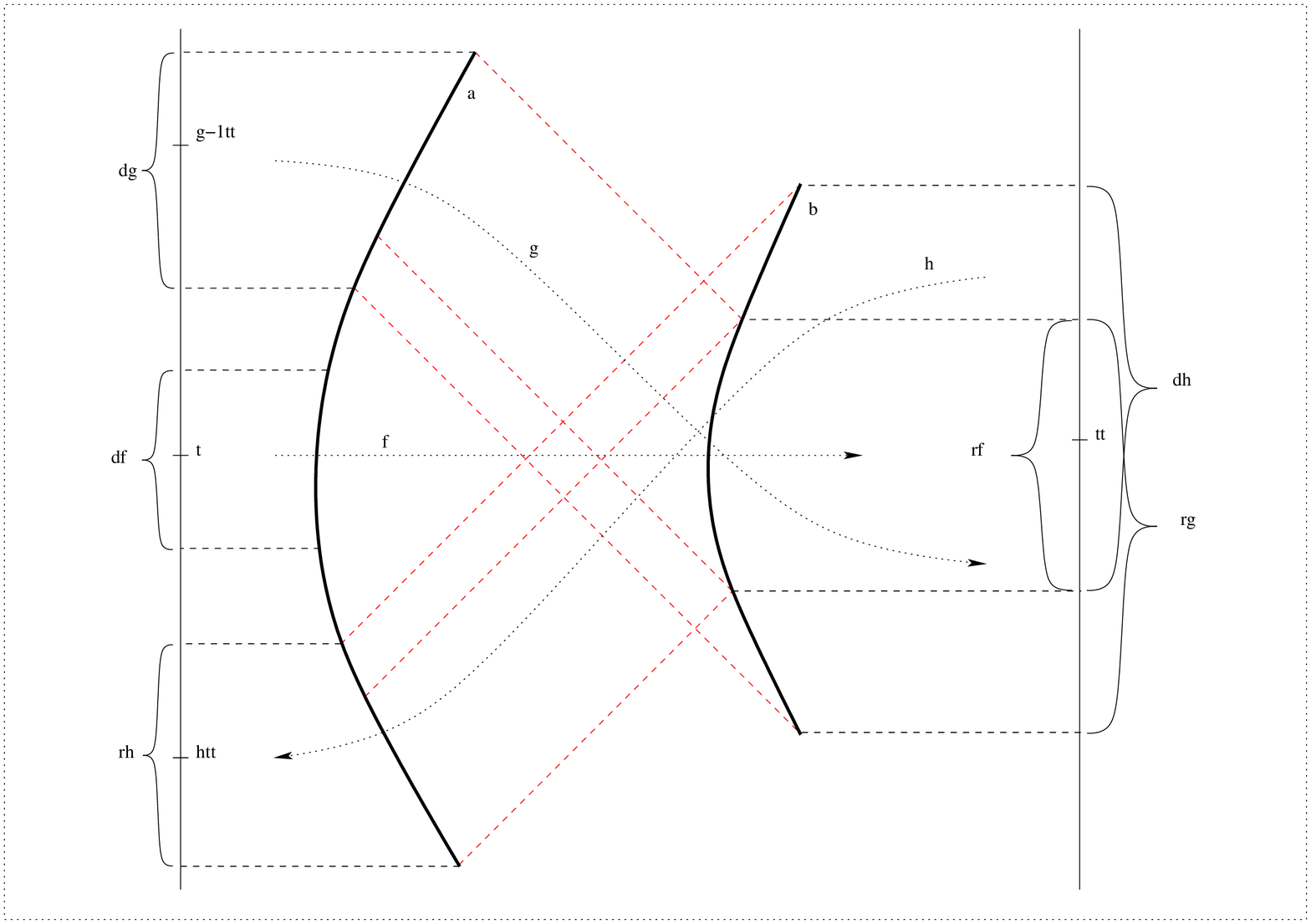}
\caption{\label{fig-radpar} Illustration for the proof of Proposition \ref{prop-rad}}
\end{center}
\end{figure}

Without losing generality, we can assume that $\alpha_\tau$ is increasing, see Lemmas \ref{lem-tlnice} and \ref{lem-inj}.

To show that $\beta_*$ is a function, let $\langle t,\vpp\rangle,\langle t,\vqq\rangle \in\beta_*$.
Then there are $r,s\in \Q$ such that $\vpp\in \Lambda^-_{\alpha(t+r)}\cap\Lambda^+_{\alpha(t-r)}$ and $\vqq\in \Lambda^-_{\alpha(t+s)}\cap\Lambda^+_{\alpha(t-s)}$.
We can assume that $0\le r\le s$.
Since both $\alpha$ and $\beta$ are timelike curves, $\vp=\vq$ iff $r=s$.
Therefore, if $\vp\neq\vq$, $\alpha(t+r)\ll\alpha(t+s)$ and $\alpha(t-s)\ll\alpha(t-r)$.
Thus $\vq\not\in I^-_{\vp}$ since $I^-_{\vpp}\subset I^-_{\alpha(t+r)}$ and $I^-_{\alpha(t+r)}\cap\Lambda^-_{\alpha(t+s)}=\emptyset$;
and $\vq\not\in I^+_{\vp}$ since $I^+_{\vpp}\subset I^+_{\alpha(t-r)}$ and $I^+_{\alpha(t-r)}\cap\Lambda^+_{\alpha(t-s)}=\emptyset$.
Thus $\vpp=\vqq$ since $\vqq\in I_{\vpp}$ by Lemma \ref{lem-chord}.

For all $t\in\dom\beta_*$, let $\tilde{t}\in\dom\beta$ such that $\beta(\tilde{t}\,)=\beta_*(t)$, and let $f:t \mapsto \tilde{t}$ be the (radar )reparametrization map, that is, $f\leteq\beta_*\circ\beta^{-1}$.
Then $f$ is injective since if $\Lambda^-_{\alpha(t_1+r)}\cap\Lambda^+_{\alpha(t_1-r)}\cap\Lambda^-_{\alpha(t_2+s)}\cap\Lambda^+_{\alpha(t_2-s)}\neq \emptyset$, then $t_1=t_2$ and $r=s$, see \eqref{item-distjcones} in Lemma \ref{lem-cau}.
Let $g$ and $h$ be the photon reparametrization maps of $\beta$ according to $\alpha$ and of $\alpha$ according to $\beta$, respectively.
Then $g$, $g^{-1}$ and $h$, $h^{-1}$ are monotonous and continuous bijections between connected sets, see Proposition \ref{prop-ph} and Lemma \ref{lem-injcont}.
It is clear by definition, that
\begin{equation}
f^{-1}(\tilde{t}\,)=t=\frac{g^{-1}(\tilde{t}\,)+h(\tilde{t}\,)}{2}
\end{equation}
for all $\tilde{t}\in\ran f$, see Figure \ref{fig-radpar}.
Thus $f^{-1}$ is continuous since both $h$ and $g^{-1}$ are so.
It is clear that $\dom f^{-1}=\ran f=\dom h \cap \ran g$.
Thus $\dom f^{-1}$ is connected since both $\dom h$ and $\ran g$ are such.
Therefore, $\dom\beta_*=\dom f=\ran f^{-1}$ is also connected and $f$ is definable and continuous, see Lemma \ref{lem-injcont}.
Hence $\beta_*=f\circ\beta$ is also continuous;
and $\beta_*$ is injective since both $\beta$ and $f$ are such.
So Item \eqref{item-radcont} is proved.

Now let us prove Item \eqref{item-raddiff}.
If $r=0$, then $\beta_*$ is the restriction of $\alpha$ to $\dom\beta_*$ which is connected, thus it is obviously differentiable.
If $r\neq0$, then $\ran \alpha\cap \ran\beta=\emptyset$.
Thus by \eqref{item-phder} in Proposition \ref{prop-ph} and Lemma \ref{lem-repar}, we have that $h$ and $g^{-1}$ are differentiable.
Thus $f$ is also differentiable.

To prove Item \eqref{item-radder}, let $\ran \alpha\cup \ran\beta\subset Plane(t,x)$.
By Item \eqref{item-raddiff} of this Proposition, $\beta_*$ is differentiable.
It is easy to see that 
\begin{equation}\label{eq-radpar}
\begin{split}
&\beta_*(t)=\alpha(t-r)\phsum \alpha(t+r) \text{ iff } \big(\beta_*(t)-\alpha(t)\big)\upp\vex \text{ and}\\
&\beta_*(t)=\alpha(t+r)\phsum\alpha(t-r) \text{ iff } \big(\beta_*(t)-\alpha(t)\big)\upp -\vex
\end{split}
\end{equation}
if $t\in\dom\beta_*$ since $\beta$ is constant $r$ radar distance from $\alpha$.
By Lemma \ref{lem-twocurve}, we have that the direction of $\beta_*(t)-\alpha(t)$ cannot change.
Thus it is always the same equation in \eqref{eq-radpar} that holds for $\beta_*$.
Hence Item \eqref{item-radder} follows by an easy calculation from Lemma \ref{lem-phsum}.
\end{proof}

If $\alpha:\Q\parrow\Q^d$ and $p\in\Q^d$, we abbreviate $\alpha(t)\upp\vp$ for all $t\in\dom\alpha$ to $\alpha\upp\vp$.
We use analogously the notation $\alpha\upp\beta$ if $\alpha,\beta:\Q\parrow\Q^d$.
Let $\Df{\bar{\alpha}}\leteq\langle \alpha_2,\alpha_1,\alpha_3,\ldots,\alpha_d\rangle$ for all $\alpha: \Q\rightarrow\Q^d$, that is, the first two coordinates are interchanged.
Let us also use the following notation
\begin{equation}
\Df{\vet}=\langle 1,0,0,\ldots,0\rangle.
\end{equation}

\begin{lem}
 \label{lem-accdir}
Assume \ax{IND}.
Let $\alpha$ be a definable timelike curve.
\begin{enumerate}
\item \label{item-velocdir} Then $\alpha'\upp\vet$ or $\alpha'\upp-\vet$.
\item \label{item-condir} If $\ran\alpha\subset Plane(t,x)$, then $\bar{\alpha}'\upp\vex$ iff $\alpha'\upp\vet$ and $\bar{\alpha}'\upp-\vex$ iff $\alpha'\upp-\vet$.
\item \label{item-accdir} If $\alpha$ is twice differentiable, $\ran\alpha\subset Plane(t,x)$ and $\vo\not\in\ran\alpha''$,
then $\alpha''\upp\vex$ ($\alpha''\upp-\vex$) iff $\alpha'_2$ is increasing (decreasing).
\item \label{item-accsame} If $\alpha$ is twice differentiable, $\ran\alpha$ is in a vertical plane and $\vo\not\in\ran\alpha''$, then $\alpha''(t_1)\upp\alpha''(t_2)$ for all $t_1,t_2\in \dom\alpha$.
\item \label{item-conacc} If $\alpha$ is twice differentiable and $\ran\alpha\subset Plane(t,x)$, then for all $t\in\dom\alpha$, there is a $\lambda_t\in\Q$ such that $\lambda_t\alpha'(t)=\alpha''(t)$.
Furthermore, if $\vo\not\in\ran\alpha''$, the sign of $\lambda_t$ is the same for all $t\in\dom\alpha$ and
\begin{equation}\label{eq-accdir}
\begin{split}
& \lambda_t>0 \quad\text{ iff }\enskip \mbox{$\phantom{-}\bar\alpha'\upp\alpha''$}\\
&\lambda_t<0 \quad\text{ iff }\enskip \mbox{$-\bar\alpha'\upp\alpha''$}
\end{split}
\end{equation}
\end{enumerate}
\end{lem}

\begin{proof}
Item \eqref{item-velocdir} is easy since by Lemma \ref{lem-tlnice}, $0\not\in\ran\alpha_\tau$.
Thus by \ax{IND}-Darboux's Theorem, we have that $\alpha'_\tau>0$ or $\alpha'_\tau<0$.

To prove Item \eqref{item-condir}, let us first note that $\alpha=\langle\alpha_\tau,\alpha_2,0,\ldots,0\rangle$ since $\ran\alpha\subset Plane(t,x)$.
Therefore, $\bar\alpha'=\langle\alpha'_2,\alpha'_\tau,0,\ldots,0\rangle$.
Hence $\bar\alpha'\upp\vex$ iff $\alpha'_\tau>0$, and $\bar\alpha'\upp-\vex$ iff $\alpha'_\tau<0$.

To prove Item \eqref{item-accdir}, let $t\in\dom\alpha$.
It is clear that $\alpha''(t)$ is spacelike or $\vo$ since $\alpha''(t)\mort\alpha'(t)$ by Proposition \ref{prop-vmorta}.
Thus $\vo\not\in\ran\alpha''$ iff $\vo\not\in\ran\alpha''_\sigma$.
We have that $\alpha_\sigma=\langle\alpha_2,0,\ldots,0\rangle\in\Q^{d-1}$ since $\ran\alpha\subset Plane(t,x)$.
Thus $\vo\not\in\ran\alpha''_\sigma$ iff $0\not\in\ran\alpha''_2$.
Hence $0\not\in\ran\alpha''_2$.
Therefore, by \ax{IND}-Darboux's theorem, we have that $\alpha''_2>0$ or $\alpha''_2<0$.
Consequently, $\alpha''\upp\vex$ iff $\alpha''_2>0$, and $\alpha''\upp-\vex$ iff $\alpha''_2<0$.
Thus since $0\not\in\ran\alpha''_2$, $\alpha''\upp\vex$ iff $\alpha'_2$ is increasing, and $\alpha''\upp-\vex$ iff $\alpha'_2$ is decreasing.

Let us now prove Item \eqref{item-accsame}.
Without losing generality, we can assume that the vertical plane is $Plane(t,x)$.
By Lemma \ref{lem-vmon}, we have that $\alpha'_2$ is increasing or decreasing since $\alpha''\circ\mu<0$ iff $\vo\not\in\ran\alpha''$.
Thus Item \eqref{item-accsame} follows by Item \eqref{item-accdir}.

Let us finally prove Item \eqref{item-conacc}.
Since both $\bar\alpha'(t)$ and $\alpha''(t)$ are Minkowski orthogonal to $\alpha'(t)$ and are in $Plane(t,x)$, there is a $\lambda_t\in\Q$ such that $\bar\alpha'(t)=\lambda_t\alpha''(t)$.
By Items \eqref{item-condir} and \eqref{item-accdir}, equation \eqref{eq-accdir} is clear.
\end{proof}

Let $\alpha$ and $\beta$ be timelike curves.
We say that $\beta_*$ is the \df{Minkowski reparametrization of $\beta$ according to $\alpha$} if 
\begin{equation}
\beta_*=\setopen \langle t,\vpp\rangle\in\dom\alpha \times\ran\beta \::\:\big(\vpp-\alpha(t)\big)\mort \alpha'(t)\setclose.
\end{equation}
We say that \df{$\beta$ is at constant $r\in\Q^+$ Minkowski distance from $\alpha$} iff for all $\vp\in\ran\beta$, there is a $t\in \dom\alpha$ such that $-\mu\big(\vp,\alpha(t)\big)=r$.

\begin{prop}
\label{prop-mink}
Assume \ax{IND}.
Let $\alpha$ and $\beta$ be definable timelike curves such that $\alpha$ is well-parametrized, and let $\beta_*$ be the Minkowski reparametrization of $\beta$ according to $\alpha$ such that.
\begin{itemize}
\item[(i)]$ \alpha$ is twice differentiable, and $\vo\not\in\ran\alpha''$.
\item[(ii)] $\ran \alpha\cup \ran \beta$ is in a vertical plane.
\item[(iii)] If $\langle t,\vpp\rangle\in\beta_*$ and $\big(\alpha(t)-\vpp\big)\upp \alpha''(t)$, then $-\mu\big(\vp,\alpha(t)\big)<-1/\mu(\alpha''(\tau))$ for all $\tau\in\dom \alpha$.\item[(iv)] $\beta$ is at constant $r\in\Q^+$ Minkowski distance from $\alpha$.
\end{itemize}
Then $\beta_*$ is a definable timelike curve.
Furthermore,
\begin{equation}\label{eq-muder}
\begin{split}
&\beta'_*(t)=\alpha'(t)+ r\cdot\bar{\alpha}''(t)\text{ iff } \alpha''(t)\upp \big(\beta_*(t)-\alpha(t)\big),\\
&\beta'_*(t)=\alpha'(t)- r\cdot\bar{\alpha}''(t)\text{ iff } \alpha''(t)\upp \big(\alpha(t)-\beta_*(t)\big)
\end{split}
\end{equation}
 if $\ran \alpha\cup \ran \beta\subseteq Plane(t,x)$, $\alpha'\upp\vet$ and $\alpha''\upp\vex$.
\end{prop}

\begin{proof}
It is clear that $\beta_*$ is definable.

To see that $\beta_*$ is a function, let $\langle t,\vqq\rangle,\langle t,\vpp\rangle\in\beta_*$.
Then $(\vpp-\vqq)\mort\alpha'(t)$.
If $\vpp\neq\vqq$, they are timelike-separated by Lemma \ref{lem-chord} since $\vpp,\vqq\in\ran\beta$.
Thus since two timelike vectors cannot be Minkowski orthogonal, we have that $\vpp=\vqq$.
Hence $\beta_*$ is a function.

Without losing generality, we can assume that the vertical plane that contains $\ran \alpha\cup \ran \beta$ is $Plane(t,x)$, $\alpha'\upp\vet$ and $\alpha''\upp\vex$, see Lemmas \ref{lem-vmon} and \ref{lem-accdir}.

Since $\beta$ is at constant $r$ Minkowski distance from $\alpha$, 
\begin{equation}\label{eq-mupar}
\begin{split}
&\beta_*(t)=\alpha(t)+r\cdot\bar\alpha'(t) \text{ iff } \bar\alpha'(t)\upp \big(\beta_*(t)-\alpha(t)\big),\\ 
&\beta_*(t)=\alpha(t)-r\cdot\bar\alpha'(t) \text{ iff } \bar\alpha'(t)\upp \big(\alpha(t)-\beta_*(t)\big)
\end{split}
\end{equation}
if $t\in\dom\beta_*$.

Since $\beta$ is at constant $r>0$ Minkowski distance from $\alpha$, we have that $\ran\alpha\cap\ran\beta=\emptyset$.
Hence by Lemma \ref{lem-twocurve}, we have that the direction of $\beta_*(t)-\alpha(t)$ cannot change.
Thus it is always the same equation in \eqref{eq-mupar} that holds for $\beta_*$.

Since $\alpha$ is twice differentiable, so is $\bar\alpha$.
Thus both $\alpha+r\cdot\bar\alpha'$ and $\alpha-r\cdot\bar\alpha'$ are definable differentiable curves.

Now we will show that $\alpha+r\cdot\bar\alpha'$ is a timelike curve and if $\bar\alpha'(t)\upp \big(\alpha(t)-\beta_*(t)\big)$ for some $t\in\dom\beta_*$, then $\alpha-r\cdot\bar{\alpha}'$ is also a timelike curve.
It is clear that $(\alpha\pm r\cdot\bar\alpha')'=\alpha'\pm r\cdot\bar\alpha''$.
Let $t\in\dom\alpha$.
By (5) in Lemma \ref{lem-accdir}, we have that $\mu\big(\alpha'(t)+r\cdot\bar\alpha''(t)\big)=\mu\big(\alpha'(t)\big)+r\mu\big(\bar\alpha''(t)\big)$ and $\mu\big(\alpha'(t)-r\cdot\bar\alpha''(t)\big)=\mu\big(\alpha'(t)\big)-r\mu\big(\bar\alpha''(t)\big)$.
By Proposition \ref{prop-wellpar}, we have that $\mu\big(\alpha'(t)\big)=1$.
Thus $\mu\big((\alpha+ r\cdot\bar\alpha')'(t)\big)>0$.
Hence $\alpha+r\cdot\bar\alpha'$ is a timelike curve.
Since $\alpha'\upp\vet$ and $\alpha''\upp\vex$, we have that $\alpha''(t)\upp\bar{\alpha}'(t)$ by Lemma \ref{lem-accdir}.
Thus by assumption (iii) and the fact that $\beta$ is at constant $r$ Minkowski distance from $\alpha$, we have that 
$r<-1/\mu(\alpha''(\tau))$ for all $\tau\in\dom\alpha$ if $\bar\alpha'(t)\upp\big(\alpha(t)-\beta_*(t)\big)$ for some $t\in\dom\alpha$.
Since $\ran \alpha\subseteq Plane(t,x)$, we have that $\mu(\alpha''(t))=-\mu(\bar\alpha''(t))$.
Thus $\mu(\bar{\alpha}''(t))<1/r$.
Consequently, $\mu\big(\alpha'(t)- r\cdot\bar\alpha''(t)\big)>0$.
Hence $\alpha-r\cdot\bar\alpha'$ is also a timelike curve.

Here we only prove that $Dom\beta_*$ is connected when $\bar\alpha'(t)\upp\big(\alpha(t)-\beta_*(t)\big)$ for some $t\in\dom\beta_*$ because the proof in the other case is almost the same.
Let $t_1,t_2\in \dom\beta_*$, and let $t\in(t_1,t_2)$.
Then $t_1,t_2\in\dom\alpha$, and thus $t\in\dom\alpha$ since $\dom\alpha$ is connected.
Since $\alpha-r\cdot\bar\alpha'$ is a timelike curve and $\alpha'-r\cdot\bar\alpha''\upp\vet$, we have that 
\begin{equation}
\begin{split}
\beta_*(t_1)=\alpha(t_1)-r\cdot\bar{\alpha}'(t_1)&\ll\alpha(t)-r\cdot\bar\alpha'(t)\\
&\ll\alpha(t_2)-r\cdot\bar{\alpha}'(t_2)=\beta_*(t_2).
\end{split}
\end{equation}
Thus by \ax{IND}-Bolzano's theorem, there is a $\bar{t}\in\dom\beta$ such that $\big(\beta(\bar{t}\,)-\alpha(t)\big)\mort \alpha'(t)$.
Since $\beta$ is at constant $r$ Minkowski distance from $\alpha$, we have that $\beta(\bar{t}\,)=\alpha(t)-r\cdot\bar\alpha'(t)$.
Hence $t\in\dom\beta_*$, as desired.

Since $\beta_*$ agrees with one of the two timelike curves $\alpha+r\cdot\bar\alpha'$ and $\alpha-r\cdot\bar\alpha'$ on the connected set $\dom\beta_*$, we have that $\beta_*$ is also a timelike curve.
Since $\alpha''\upp\vex$ and $\alpha\upp\vet$ we have that $\alpha''\upp\bar{\alpha}'$.
Thus by derivation of the equations of \eqref{eq-mupar}, we have that the derivate of $\beta_*$ is what was stated in \eqref{eq-muder}.
\end{proof}

\begin{lem}
\label{lem-time}
Assume \ax{AxSelf^-}, \ax{AxPh_0}, and let $m\in\IOb$.
Let $k\in \Ob$.
Let $x,y\in\dom lc^k_m$.
Then
\begin{equation}\label{eq-time}
\time_k\big(ev_m\big(lc^k_m(x)\big),ev_m\big(lc^k_m(y)\big)\big)=|x-y|.
\end{equation}
\end{lem}

\begin{proof}
By \eqref{item-trfunct} in Proposition \ref{prop-tr}, $lc^k_m$ is a function.
Thus $lc^k_m(x)$ and $lc^k_m(y)$ are meaningful.
We have that $k\in ev_m\big(lc^k_m(x)\big)\bigcap ev_m\big(lc^k_m(y)\big)$ by the definition of $lc^k_m$.
Thus by \ax{AxSelf^-}, both events $ev_m\big(lc^k_m(x)\big)$ and $ev_m\big(lc^k_m(y)\big)$ have unique coordinates in $Cd_k$.
Thus the left hand side of equation \eqref{eq-time} is defined and equal with 
\begin{equation}
\left|Crd_k\big(ev_m\big(lc^k_m(x)\big)\big)_\tau-Crd_k\big(ev_m\big(lc^k_m(y)\big)\big)_\tau \right|
\end{equation} 
by definition.
However, by the definition of $lc^k_m$, $Crd_k\big(ev_m\big(lc^k_m(x)\big)\big)_\tau=x$ and $Crd_k\big(ev_m\big(lc^k_m(y)\big)\big)_\tau=y$.
Hence equation \eqref{eq-time} holds.
\end{proof} 


\begin{figure}[h!btp]
\small
\begin{center}
\psfrag{g}[b][b]{$\gamma=lc^c_m$}
\psfrag{b}[l][l]{$\beta=\beta_*=lc^b_m$}
\psfrag{b1}[bl][bl]{$\beta'(t)$}
\psfrag{bt}[l][l]{$\beta(t)$}
\psfrag{g1}[tl][tl]{$\gamma'(\bar{t}\,)$}
\psfrag{g1*}[tl][tl]{$\gamma_*'(t)$}
\psfrag{gt}[tl][tl]{$\gamma(\bar{t}\,)$}
\includegraphics[keepaspectratio, width=\textwidth]{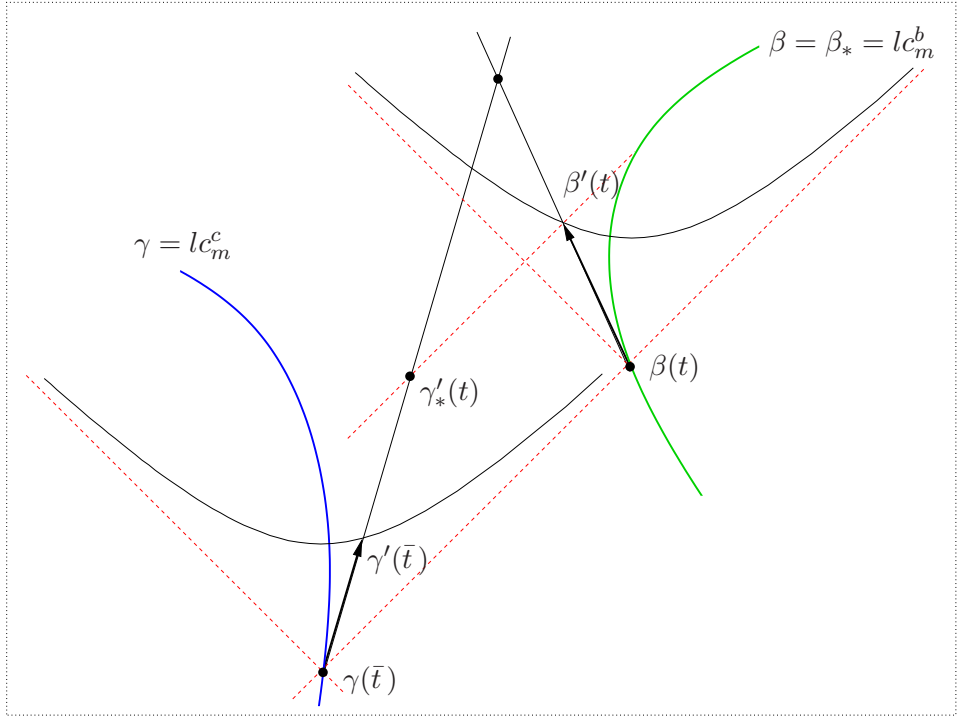}
\caption{\label{fig-thmph} Illustration for the proof of Theorem \ref{thm-ph}.}
\end{center}
\end{figure}

\begin{proof}[\colorbox{proofbgcolor}{\textcolor{proofcolor}{proof of Theorem \ref{thm-ph}}}]
To prove Item (1), let $b$ and $c$ be coplanar observers, and let $e_b$, $\bar{e}_b$, $e_c$ and $\bar{e}_c$ be such events that $b\in e_b\cap\bar{e}_b$, $c\in e_c\cap \bar{e}_c$ and $e_b\simph_b e_c$, $\bar{e}_b\simph_b\bar{e}_c$.
Suppose that $c$ is approaching $b$ as seen by $b$ by photons.
We have to prove that $\time_b(e_b,\bar{e}_b)<\time_c(e_c,\bar{e}_c)$.
Since $c$ and $b$ are coplanar, there is an $m\in\IOb$ such that $wl_m(c)\cup wl_m(b)$ is a subset of a vertical plane.
Let $m$ be such an inertial observer.
We are going to apply Lemma \ref{lem-main}.
To do so, let $\beta=\beta_*=lc^b_m$, $\gamma=lc^c_m$, and let $\gamma_*$ be the photon reparametrization of $\gamma$ according to $\beta$.
By Proposition \ref{prop-wellpar}, $\beta=\beta_*$ and $\gamma$ are definable and well-parametrized timelike curves.
Without losing generality, we can assume that $\beta'\upp\vet$ and $\gamma'\upp\vet$.
It is easy to see that $wl_m(b)\cap wl_m(c)=\emptyset$ since $c$ is approaching $b$ as seen by $b$.
Thus $\ran\beta\cap\ran\gamma=\emptyset$ since $\ran\beta=wl_m(b)$ and $\ran\gamma=wl_m(c)$ by Item \eqref{item-rantr} in Proposition \ref{prop-tr}.
Thus $\gamma_*$ is also a definable timelike curve by Proposition \ref{prop-ph}.
Requirement (i) in Lemma \ref{lem-main} is clear by the definition of the photon reparametrization.
It is also clear that there are $x_\beta,y_\beta\in\dom\beta$, $x_\gamma,y_\gamma\in \dom\gamma$ and $x,y\in\dom\beta_*\cap\dom\gamma_*$ such that 
$\beta(x_\beta)=Crd_m(e_b)=\beta_*(x)$, $\beta(y_\beta)=Crd_m(\bar{e}_b)=\beta_*(y)$ and 
$\gamma(x_\gamma)=Crd_m(e_c)=\gamma_*(x)$, $\gamma(y_\gamma)=Crd_m(\bar{e}_c)=\gamma_*(y)$.
Hence requirement (ii) in Lemma \ref{lem-main} also holds.
Since $c$ is approaching $b$ as seen by $b$ by photons, the tangent lines of $\beta_*$ and $\gamma_*$ at any $t\in(x,y)$ intersect in the future of $\beta_*(t)$ and $\gamma_*(t)$.
Thus $\mu\big(\beta'_*(t)\big)=1<\mu\big(\gamma'_*(t)\big)$ for all $t\in (x,y)$ by Proposition \ref{prop-ph}, see Figure \ref{fig-thmph};
and this is requirement (iii) in Lemma \ref{lem-main}.
Hence by Lemma \ref{lem-main}, we have that $|x_\beta-y_\beta|<|x_\gamma-y_\gamma|$.
Consequently, $\time_b(e_b,\bar{e}_b)<\time_c(e_c,\bar{e}_c)$ since by Lemma \ref{lem-time}, $\time_i(e_i,\bar{e}_i)=|x_i-y_i|$ for all $i\in\setopen b,c\setclose$.
So Item (1) is proved.

The proof of (2) is similar.
Hence it is left to the reader.
\end{proof}

\begin{figure}[h!btp]
\small
\begin{center} 
\psfrag{a}[l][l]{$\alpha$}
\psfrag{b}[l][l]{$\beta$}
\psfrag{c}[l][l]{$\gamma$}
\psfrag{at}[l][l]{$\alpha(t)$}
\psfrag{at+R}[r][r]{$\alpha(t+R)$}
\psfrag{at-R}[r][r]{$\alpha(t-R)$}
\psfrag{at+r}[r][r]{$\alpha(t+r)$}
\psfrag{at-r}[r][r]{$\alpha'(t-r)$}
\psfrag{a1t+R}[b][b]{$\alpha'(t+R)$}
\psfrag{a1t-R}[b][b]{$\alpha'(t-R)$}
\psfrag{a1t+r}[b][b]{$\alpha'(t+r)$}
\psfrag{a1t-r}[b][b]{$\alpha'(t-r)$}
\psfrag{a1t+rtl}[t][t]{$\alpha'(t+r)$}
\psfrag{a1t-rtr}[t][t]{$\alpha'(t-r)$}
\psfrag{bt}[l][l]{$\beta_*(t)$}
\psfrag{btr}[r][r]{$\beta_*(t)$}
\psfrag{atr}[r][r]{$\alpha(t)$}
\psfrag{ct}[l][l]{$\gamma_*(t)$}
\psfrag{a1t}[b][b]{$\alpha'(t)$}
\psfrag{b1t}[l][l]{$\beta'_*(t)$}
\psfrag{c1t}[l][l]{$\gamma'_*(t)$}
\psfrag{text}[l][l]{$\alpha_2(x)<\alpha_2(y) \text{ iff } x<y$}
\includegraphics[keepaspectratio, width=\textwidth]{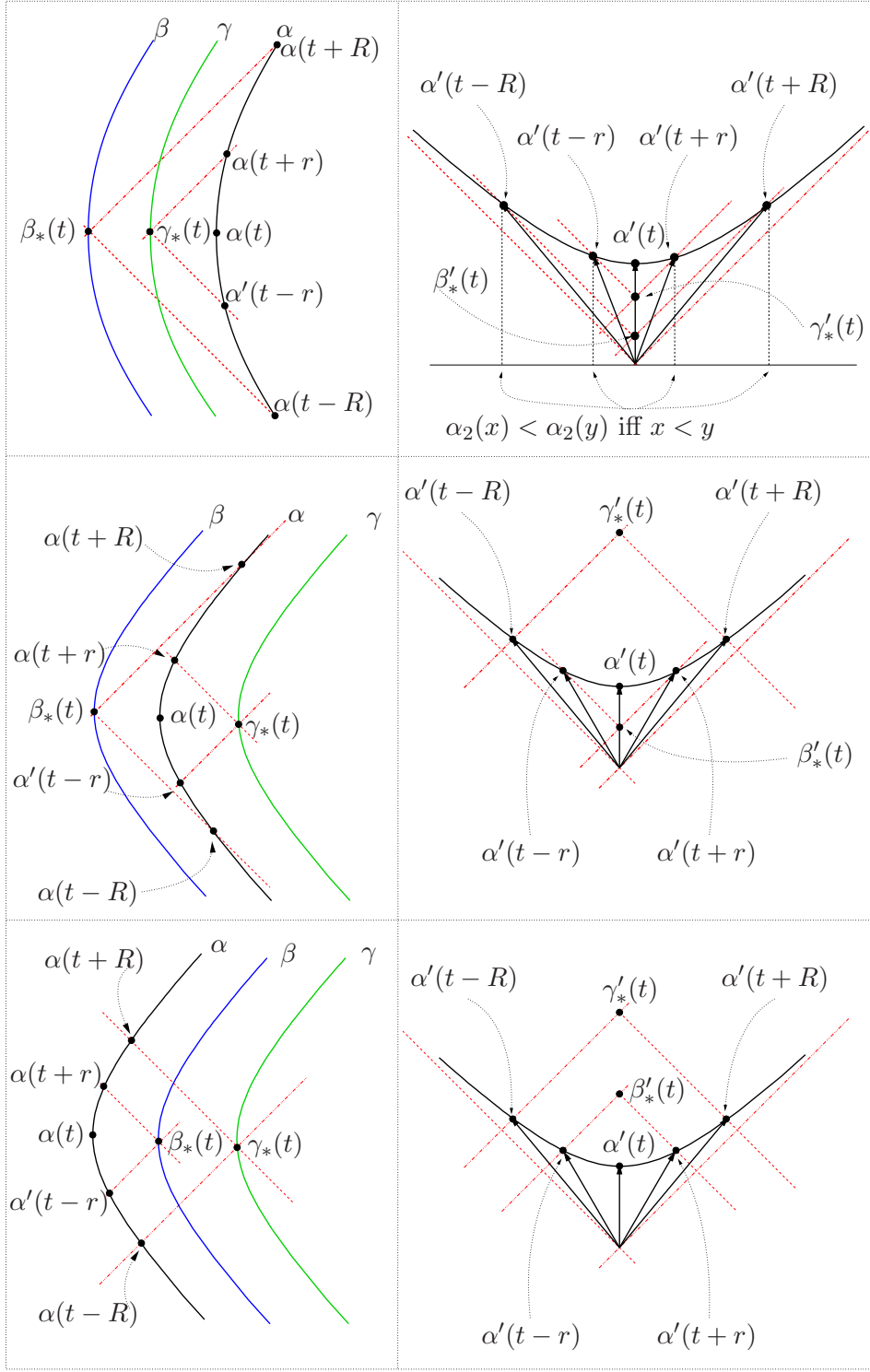}
\caption{\label{fig-radthm} Illustration for the proof of Item (1) in Theorem \ref{thm-rad}
verifying requirement (iii) in Lemma \ref{lem-main}.}
\end{center}
\end{figure}

\begin{proof}[\colorbox{proofbgcolor}{\textcolor{proofcolor}{proof of Theorem \ref{thm-rad}}}]
To prove Item (1), let $\rship$ be a radar spaceship such that $k$ is positively accelerated and the direction of the spaceship is the same as that of the acceleration of $k$.
Let $e_b$, $\bar{e}_b$, $e_c$, $\bar{e}_c$ be such events that $b\in e_b\cap\bar{e}_b$, $c\in e_c\cap \bar{e}_c$ and $e_b\simrad_k e_c$, $\bar{e}_b\simrad_k\bar{e}_c$.
To prove Item (1), we have to prove that $\time_b(e_b,\bar{e}_b)<\time_c(e_c,\bar{e}_c)$.
Since $\rship$ is a spaceship, there is an $m\in\IOb$ such that $wl_m(b)\cup wl_m(k)\cup wl_m(c)$ is a subset of a vertical plane.
Let $m$ be such an inertial observer.
Without losing generality, we can assume that this plane is $Plane(t,x)$.
We are going to apply Lemma \ref{lem-main}.
To do so, let $\beta=lc^b_m$, $\gamma=lc^c_m$ and $\alpha=lc^k_m$;
and let $\beta_*$ and $\gamma_*$ be the radar reparametrization of $\beta$ and $\gamma$ according to $\alpha$, respectively.
By Proposition \ref{prop-wellpar}, $\beta$ and $\gamma$ are definable and well-parametrized timelike curves.
By Lemmas \ref{lem-vmon} and \ref{lem-accdir}, we can assume that $\alpha'_2$ is increasing and $\alpha'\upp\vet$.
By Proposition \ref{prop-ph}, $\beta_*$ and $\gamma_*$ are definable timelike curves since the photon sum of any two timelike vectors of $\ran\alpha'$ is also a timelike one.
Requirement (i) in Lemma \ref{lem-main} is clear by the definition of the radar reparametrization.
It is also clear that there are $x_\beta,y_\beta\in\dom\beta$, $x_\gamma,y_\gamma\in \dom\gamma$ and $x,y\in\dom\beta_*\cap\dom\gamma_*$ such that 
$\beta(x_\beta)=Crd_m(e_b)=\beta_*(x)$, $\beta(y_\beta)=Crd_m(\bar{e}_b)=\beta_*(y)$ and 
$\gamma(x_\gamma)=Crd_m(e_c)=\gamma_*(x)$, $\gamma(y_\gamma)=Crd_m(\bar{e}_c)=\gamma_*(y)$.
Hence requirement (ii) in Lemma \ref{lem-main} also holds.
Since the direction of $\rship$ is the same as that of the acceleration of $k$, there are only three possible orders of the observers in the spaceship.
All these three cases are illustrated by Figure \ref{fig-radthm}.
By Proposition \ref{prop-rad}, it is easy to see that $\mu\big(\beta'_*(t)\big)<\mu\big(\gamma'_*(t)\big)$ for all $t\in (x,y)$;
and this is requirement (iii) in Lemma \ref{lem-main}.
Hence by Lemma \ref{lem-main}, $|x_\beta-y_\beta|<|x_\gamma-y_\gamma|$.
Thus $\time_b(e_b,\bar{e}_b)<\time_c(e_c,\bar{e}_c)$ since by Lemma \ref{lem-time}, $\time_i(e_i,\bar{e}_i)=|x_i-y_i|$ for all $i\in\setopen b,c\setclose$;
and this is what we wanted to prove.

\begin{figure}[h!btp]
\small
\begin{center} 
\psfrag{a}[bl][b]{$\alpha$}
\psfrag{b}[bl][bl]{$\beta$}
\psfrag{c}[bl][bl]{$\gamma$}
\psfrag{at}[tr][tr]{$\alpha(t)$}
\psfrag{at-2R}[tl][tl]{$\alpha(t-2R)$}
\psfrag{at-2r}[tl][tl]{$\alpha(t-2r)$}
\psfrag{a1t-2R}[tr][tl]{$\alpha'(t-2R)$}
\psfrag{a1t-2r}[tr][tl]{$\alpha'(t-2r)$}
\psfrag{b*t}[tr][tr]{$\beta_*(t)$}
\psfrag{c*t}[tl][tl]{$\gamma_*(t)$}
\psfrag{b1*t}[tr][tl]{$\beta'_*(t)$}
\psfrag{c1*t}[tl][tl]{$\gamma'_*(t)$}
\psfrag{a1t}[tl][tl]{$\alpha'(t)$}
\psfrag{b*1t}[tl][tl]{$\beta'_*(t)$}
\psfrag{c*1t}[tl][tl]{$\gamma'_*(t)$}
\includegraphics[keepaspectratio, width=\textwidth]{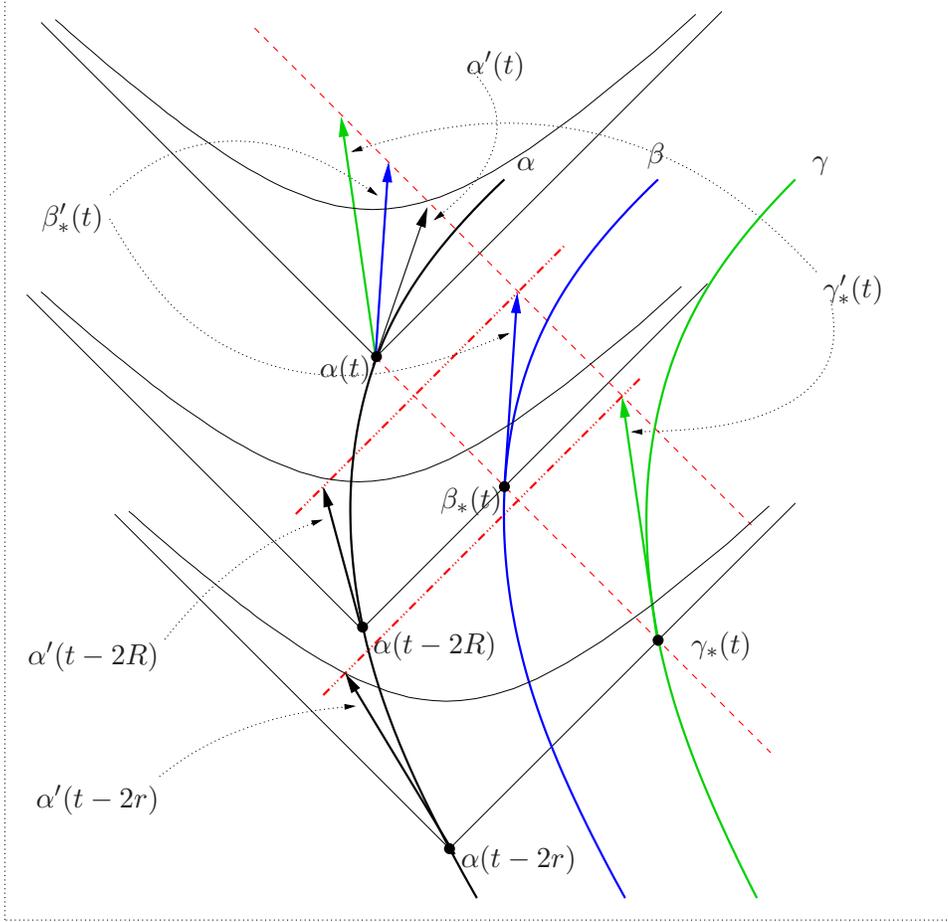}
\caption{\label{fig-radthmph} Illustration for the proof of Item (2) in Theorem \ref{thm-rad}
verifying requirement (iii) in Lemma \ref{lem-main}.}
\end{center}
\end{figure}

To prove Item (2), there are many cases we should consider resulting from which order is taken by the observers in the spaceship, and which observer is watching the other two.
The proof in all the cases is based on the very same ideas and lemmas as the proof of Item (1).
The only difference is that we should use photon simultaneity and photon reparametrization instead of radar ones, and we should use Proposition \ref{prop-ph} (and Lemma \ref{lem-vmon}) when verifying requirement (iii) in Lemma \ref{lem-main}.
In Figure \ref{fig-radthmph}, we illustrate the proof of requirement (iii) in Lemma \ref{lem-main} in one of the many cases.
In the other cases, this part of the proof can also be attained by means similar figures without any difficulty.
\end{proof}

\begin{figure}[h!btp]
\small
\begin{center} 
\psfrag{a}[bl][bl]{$\alpha$}
\psfrag{b}[bl][bl]{$\beta$}
\psfrag{c}[bl][bl]{$\gamma$}
\psfrag{aa}[tl][tl]{(a)}
\psfrag{bb}[tl][tl]{(b)}
\psfrag{cc}[tl][tl]{(c)}
\psfrag{dd}[tl][tl]{(d)}
\psfrag{a1}[b][b]{$\alpha'(t)$}
\psfrag{a1t}[tl][tl]{$\alpha'(t)$}
\psfrag{at}[tl][tl]{$\alpha(t)$}
\psfrag{a11}[tl][tl]{$\alpha''(t)$}
\psfrag{b*t}[tr][tl]{$\beta_*(t)$}
\psfrag{c*t}[bl][bl]{$\gamma_*(t)$}
\psfrag{b1*t}[tr][tl]{$\beta'_*(t)$}
\psfrag{c1*t}[tl][tl]{$\gamma'_*(t)$}
\psfrag{b*1}[bl][bl]{$\beta'_*(t)$}
\psfrag{b*1r}[br][br]{$\beta'_*(t)$}
\psfrag{c*1}[bl][bl]{$\gamma'_*(t)$}
\psfrag{r}[bl][bl]{$r$}
\psfrag{R}[tl][tl]{$R$}
\psfrag{dc}[bl][bl]{$r\cdot\bar\alpha''(t)$}
\psfrag{db}[bl][bl]{$R\cdot\bar\alpha''(t)$}
\includegraphics[keepaspectratio, width=\textwidth]{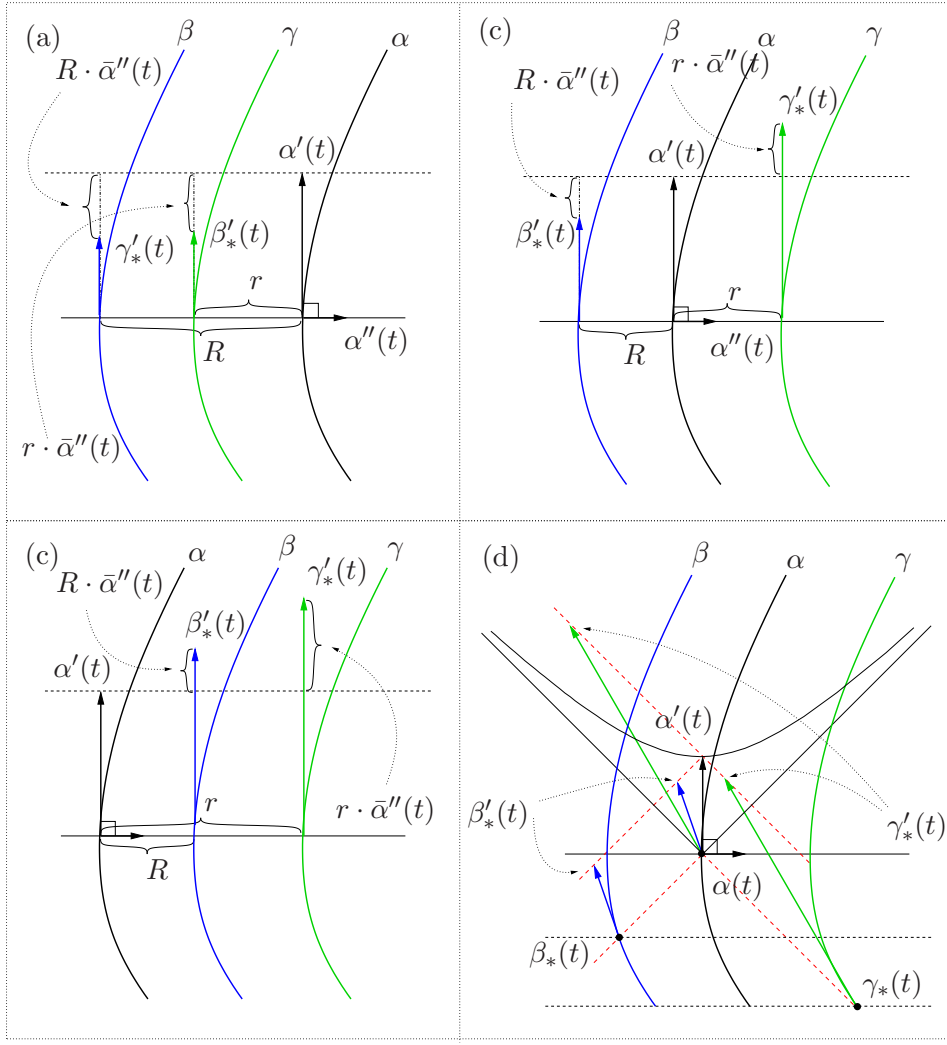}
\caption{\label{fig-minkthm} Illustration for the proof of Theorem \ref{thm-mu}
verifying requirement (iii) in Lemma \ref{lem-main}.}
\end{center}
\end{figure}

\begin{proof}[\colorbox{proofbgcolor}{\textcolor{proofcolor}{on the proof of Theorem \ref{thm-mu}}}]
The proof of this Theorem is based on the very same ideas and lemmas as the proof of Theorem \ref{thm-rad}.
The only difference is that we should use Minkowski simultaneity and Minkowski reparametrization instead of radar ones, and in the proof of Item (1) we should use Proposition \ref{prop-mink} instead of Proposition \ref{prop-rad} when verifying requirement (iii) in Lemma \ref{lem-main}.
In the proof of Item (1) of this Theorem, we face the same three cases as in the proof of Item (1) in Theorem \ref{thm-rad}.
By (a), (b) and (c) of Figure \ref{fig-minkthm}, we illustrate the proof of requirement (iii) in Lemma \ref{lem-main} in this three cases.
Similarly, in the proof of Item (2) of this Theorem, we face the same large number of cases as in the proof of Item (2) in Theorem \ref{thm-rad}.
By (d) of Figure \ref{fig-minkthm}, we illustrate the proof of requirement (iii) in Lemma \ref{lem-main} in one of these many cases.
We do not go into more details here because we think that the reader can easily put the proof together with the help of the hints above.
\end{proof}

\begin{figure}[h!btp]
\small
\begin{center}
\psfrag{be}[l][tl]{$\beta$}
\psfrag{ga}[l][tl]{$\gamma$}
\psfrag{b}[l][tl]{$b$}
\psfrag{c}[l][tl]{$c$}
\psfrag{o}[r][r]{$\vo$}
\psfrag{p}[tr][tr]{$\vp$}
\psfrag{q}[tr][tr]{$\vq$}
\psfrag{p1}[br][br]{$\vpp'$}
\psfrag{q1}[br][br]{$\vqq'$}
\psfrag{eq}[tl][tl]{$=$}
\psfrag{e}[r][r]{$e$}
\psfrag{ph}[l][bl]{$ph$}
\psfrag{eb}[tl][tl]{$e_b$}
\psfrag{ec}[tl][tl]{$e_c$}
\includegraphics[keepaspectratio, width=\textwidth]{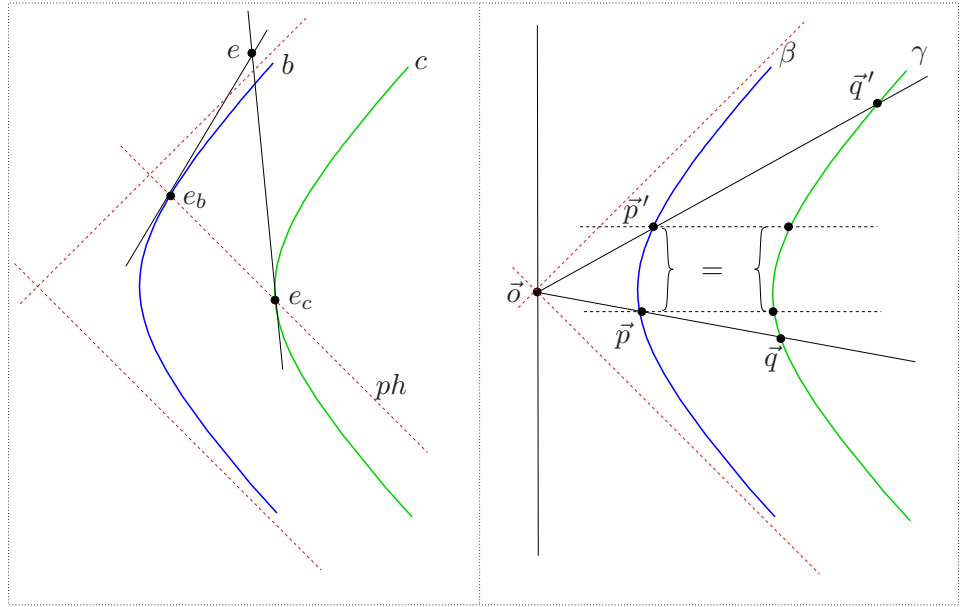}
\caption{\label{fig-unifacc} Illustration for the proof of Theorem \ref{thm-ob}.}
\end{center}
\end{figure}

\begin{proof}[\colorbox{proofbgcolor}{\textcolor{proofcolor}{on the proof of Theorem \ref{thm-ob}}}]
Let $\Q$ be the field of real numbers.
Let $\beta(t)=\big(sh(t),ch(t),0,\ldots,0\big)$ and $\gamma(t)=\big(sh(t),ch(t)+1,0\ldots,0\big)$ where $sh$ and $ch$ are the hyperbolic sine and cosine functions.
Since both $\beta$ and $\gamma$ are smooth and well-parametrized timelike curves, we can easily build a model of \ax{AccRel} such that $lc^b_m=\beta$ and $lc^c_m=\alpha$ for some $m\in\IOb$.
By a straightforward calculation, we can show that $\mu\big(\beta''(t)\big)=\mu\big(\gamma''(t)\big)=-1$ for all $t\in\Q$.
Hence $a_b(t)=a_c(t)=1$ for all $t\in\Q$.

It is easy to show that $c$ is approaching $b$ as seen by $b$ by photons, see (a) of Figure \ref{fig-unifacc}.
Thus by Theorem \ref{thm-ph}, the clock of $b$ runs slower
than the clock of $c$ as seen by $b$ by photons.
It is not difficult to show that, $ev_m(\vpp)\simrad_b ev_m(\vqq)$ iff $ev_m(\vpp)\simmu_b ev_m(\vqq)$ iff $\vo\in line(\vpp,\vqq)$.
Thus the clock of $b$ runs slower than the clock of $c$ as seen by $b$ by both radar simultaneity and Minkowski simultaneity, see (b) of Figure \ref{fig-unifacc}.
\end{proof}

\section{Lemmas from analysis generalized for FOL}
\label{lem-sec}

For the sake of completeness, here we list some of the basic definitions and theorems of real analysis generalised for ordered fields and {\em definable} functions within FOL.
For more details and proofs, see \cite{Twp,mythes}.

We call a function $f:\Q\parrow\Q^n$ \df{continuous at} $t_0\in\dom f$ iff
\begin{equation}
\begin{split}
\forall \varepsilon\in\Q^+\;\exists\delta\in\Q^+\enskip &\forall t\in\dom f\\
|t-t_0|&<\delta \then |f(t)-f(t_0)|<\varepsilon.
\end{split}
\end{equation}
We call function $f$ \df{monotonous} if it preserves or reverses the relation $<$, that is, $f(x)<f(y)$ [or $f(x)>f(y)$] for all $x,y\in \dom f$ if $x<y$.

\begin{lem}
\label{lem-moncont}
If $f:\Q\parrow\Q$ is monotonous and $\ran f$ is connected, $f$ is continuous.\qed
\end{lem}

\begin{lem}
\label{lem-injcont}
Assume \ax{IND}.
Let $f:\Q\parrow\Q$ be definable and continuous such that $\dom f$ is connected. Then
\begin{enumerate}
\item $\ran f$ is also connected.
\item If $f$ is injective, $f$ is monotonous too.
Moreover, $f^{-1}$ is also a definable monotonous and continuous function.
\end{enumerate}
\end{lem}

\begin{proof}
Item (1) is an easy consequence of \ax{IND}-Bolzano theorem.
To prove Item (2), let us first note that if $f$ were not monotonous, it would not be injective by \ax{IND}-Bolzano theorem.
It is clear that $f^{-1}$ is definable and monotonous since $f$ is such.
Thus by Lemma \ref{lem-moncont}, $f^{-1}$ is continuous.
\end{proof}

We say that a function $f:\Q\parrow\Q^n$ is \df{differentiable at} $t_0\in\dom f$ iff there is a unique $f'(t_0)\in\Q^n$ such that
\begin{equation}
\begin{split}
\forall \varepsilon\in\Q^+\;&\exists\delta\in\Q^+\enskip \forall t\in\dom f\quad |t-t_0|<\delta \\
&\then |f(t)-f(t_0)-f'(t_0)\cdot(t-t_0)|<\varepsilon\cdot|t-t_0|.
\end{split}
\end{equation}
This $f'(t_0)$ is called the \df{derivate of $f$ at $t_0$}.
Let us introduce the following convenient abbreviation.
We say that $\alpha:\Q\parrow \Q$ is a \df{nice map} if it is a differentiable such that $0\not\in\ran\alpha'$, and $\dom \alpha$ is connected.

\begin{lem}
\label{lem-tlnice}
Let $\alpha$ be a timelike curve.
Then $\alpha_\tau$ is a nice map.
\end{lem}

\begin{proof} Since $\alpha$ is a timelike curve, $\dom\alpha$ is connected and $\alpha'(x)_\tau\neq 0$ for all $x\in \dom \alpha$.
But $\dom\alpha_\tau=\dom\alpha$ and $(\alpha_\tau)'=(\alpha')_\tau$.
Thus $\alpha_\tau$ is a nice map.
\end{proof}

\begin{lem}
\label{lem-inj}
Assume \ax{IND}.
Let $\alpha$ be a definable nice map.
Then $\alpha$ is injective.
Moreover, $\alpha$ is monotonous.
\end{lem}

\begin{proof} If $\alpha$ were not injective, then $\alpha'(x)$ would be $0$ for some $x$ by \ax{IND}-Rolle's theorem.
But $\alpha'(x)$ cannot be $0$ since $\alpha$ is a nice map.
Thus $\alpha$ is injective.
Then $\alpha$ is also monotonous by (2) in Lemma \ref{lem-injcont}.
\end{proof}
 
\begin{lem}
\label{lem-nice}
Assume \ax{IND}.
If $\alpha$ and $\delta$ are nice maps, $\delta^{-1}$ and $\alpha\circ\delta$ are also nice maps.\qed
\end{lem}

\begin{lem}
\label{lem-mink} Assume \ax{IND}.
Let $\alpha$ and $\delta$ be definable timelike curves such that $\ran\alpha\subseteq \ran\delta$ (or $\ran\delta\subseteq \ran\alpha$), and let $h\leteq\alpha\circ\delta^{-1}$.
Then $h$ is a nice map and 
\begin{equation}\label{eq-mink}
|h'(x)|=\frac{\mu\big(\alpha'(x)\big)}{\mu\big(\delta'(h(x))\big)} \quad \text{ for all } \enskip x\in \dom h.
\end{equation}
\end{lem}

\begin{proof} 
First we show that $h=\alpha_\tau\circ\delta_\tau^{-1}$.
Since $\alpha$ and $\delta$ are definable timelike curves, $\alpha_\tau$ and $\delta_\tau$ are definable nice maps by Lemma \ref{lem-tlnice}.
Thus $\alpha_\tau$ and $\delta_\tau$ are injective by Lemma \ref{lem-inj}.
Consequently, $\alpha$ and $\delta$ are also injective.
Therefore, $\langle x,y\rangle \in \alpha_\tau\circ \delta_\tau^{-1}$ iff $\alpha_\tau(x)=\delta_\tau(y)$ and $\langle x,y\rangle\in \alpha\circ\delta^{-1}$ iff $\alpha(x)=\delta(y)$.
Since $\alpha(x)=\delta(y)\then \alpha_\tau(x)=\delta_\tau(y)$ is clear, we have to show the converse implication only.
By symmetry, we can assume that $\ran\alpha\subseteq \ran\delta$.
Then there is a $z\in \dom\delta$ such that $\delta(z)=\alpha(x)$, so $\delta_\tau(z)=\alpha_\tau(x)=\delta_\tau(y)$.
Thus $z=y$ since $\delta$ is injective, so $\alpha(x)=\delta(y)$.
This proves $h=\alpha_\tau\circ\delta_\tau^{-1}$.

By Lemma~\ref{lem-nice}, $h$ is a nice map, so $\dom h$ is an interval.
We have that $\alpha\supseteq h\circ\delta$ since $h=\alpha\circ\delta^{-1}$.
Thus by the chain rule, $\alpha'(x)=h'(x)\cdot\delta'\big(h(x)\big)$ for all $x\in \dom h$.
Since $\mu(\lambda\vpp)=|\lambda|\cdot\mu(\vpp)$ for all $\lambda\in Q$ and $\vpp\in\Q^d$, we have that $\mu\big(\alpha'(x)\big)=|h'(x)|\cdot\mu\big(\delta'(h(x))\big)$ for all $x\in \dom h$.
We have that $\mu\big(\delta'(h(x))\big)\neq0$ since $\delta$ is timelike.
Hence equation \eqref{eq-mink} holds.
\end{proof}

\begin{lem}
\label{lem-main}
Assume \ax{IND}.
Let $\beta$ and $\gamma$ be definable and well-parametrized timelike curves; let $\beta_*$ and $\gamma_*$ be definable timelike curves; let $x_\beta,y_\beta\in \dom\beta$, $x_\gamma,y_\gamma\in \dom\gamma$ and $x,y\in \dom\beta_*\cap \dom\gamma_*$ such that
\begin{itemize}
\item[(i)] $\ran\beta_*\subseteq \ran\beta$ and $\ran\gamma_*\subseteq \ran\gamma$.
\item[(ii)] $\beta(x_\beta)=\beta_*(x)$, $\beta(y_\beta)=\beta_*(y)$, $\gamma(x_\gamma)=\gamma_*(x)$, $\gamma(y_\gamma)=\gamma_*(y)$.
\item[(iii)] $x\neq y$ and $\mu\big(\gamma'_*(z)\big)>\mu\big(\beta'_*(z)\big)$ for all $z\in (x,y)$.
\end{itemize}
Then $\big|x_\gamma-y_\gamma\big|>\big|x_\beta-y_\beta\big|$.
\end{lem}

\begin{figure}[h!btp]
\small
\begin{center}
\psfrag{yb}[r][r]{$y_\beta$}
\psfrag{xb}[r][r]{$x_\beta$}
\psfrag{y}[r][r]{$y$}
\psfrag{x}[r][r]{$x$}
\psfrag{yc}[r][r]{$y_\gamma$}
\psfrag{xc}[r][r]{$x_\gamma$}
\psfrag{i}[t][t]{$i$}
\psfrag{j}[t][t]{$j$}
\psfrag{bb}[bl][bl]{$\beta_*,\beta$}
\psfrag{cc}[bl][bl]{$\gamma_*,\gamma$}
\psfrag{b}[b][b]{$\beta$}
\psfrag{c}[b][b]{$\gamma$}
\psfrag{bx}[b][b]{$\beta_*$}
\psfrag{cx}[b][b]{$\gamma_*$}
\psfrag{bxy}[l][l]{$$}
\psfrag{cxy}[l][l]{$$}
\psfrag{bxx}[l][l]{$$}
\psfrag{cxx}[l][l]{$$}
\includegraphics[keepaspectratio, width=\textwidth]{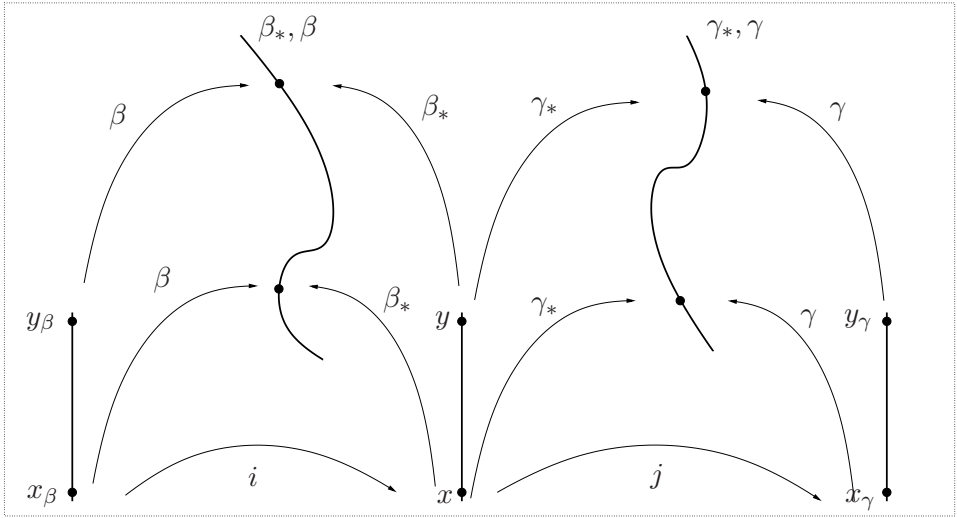}
\caption{\label{fig-lemma} Illustration for the proof of Lemma \ref{lem-main}.}
\end{center}
\end{figure}

\begin{proof}
Since $\beta$, $\beta_*$, $\gamma$ and $\gamma_*$ are definable timelike curves, they are injective by Lemmas \ref{lem-tlnice} and \ref{lem-inj}.
Thus $x_\beta\neq y_\beta$ and $x_\gamma\neq y_\gamma$ since $x\neq y$.
Let 
\begin{equation}
i\leteq\beta\circ\beta_*^{-1}\quad\text{and}\quad j\leteq\gamma_*\circ \gamma^{-1},
\end{equation} 
see Figure \ref{fig-lemma}.
Then $i$, $j$ and $i\circ j$ are nice maps by Lemma~\ref{lem-nice} and \ref{lem-mink}.
Furthermore,
\begin{equation}
\begin{split}
i(x_\beta)=x,\enskip i(y_\beta)=y,\enskip j(x)=x_\gamma,\enskip j(y)=y_\gamma,\enskip\\
(i\circ j)(x_\beta)=x_\gamma\enskip \text{and}\enskip (i\circ j)(y_\beta)=y_\gamma.
\end{split}
\end{equation}
Since $x_\beta, y_\beta\in \dom(i\circ j)$, and $i\circ j$ is a nice map, we have that $(x_\beta,y_\beta)\subseteq \dom(i\circ j)$.

Now we will show that
\begin{equation}
\label{megjelolt}
\forall t\in (x_\beta, y_\beta) \quad \big|(i\circ j)'(t)\big|>1.
\end{equation}
To prove this statement, let $t\in (x_\beta, y_\beta)$.
Since $i$ is a nice map, it is monotonous by Lemma~\ref{lem-inj}, thus $i(t)\in (x,y)$.
By Lemma~\ref{lem-mink} and the fact that $\beta$ and $\gamma$ are well-parametrized, we have that 
\begin{equation}\label{i}
\big|i'(t)\big|=\frac{\mu\big(\beta'(t)\big)}{\mu\big(\beta_*'(i(t))\big)}=\frac{1}{\mu\big(\beta_*'(i(t))\big)}\\ 
\end{equation}
and
\begin{equation}\label{j}
\big|j'\big(i(t)\big)\big|=\frac{\mu\big(\gamma_*'(i(t))\big)}{\mu\big(\gamma'\big(j(i(t))\big)\big)}=\mu\big(\gamma'_*(i(t))\big).
\end{equation}
From equations \eqref{i}, \eqref{j} and Item (iii) by the chain rule, we have that
\begin{equation}
\big|(i\circ j)'(t)\big|=\big|i'(t)j'\big(i(t)\big)\big|=\frac{\mu\big(\gamma_*'(i(t))\big)}{\mu\big(\beta'_*(i(t))\big)}>1
\end{equation}
This completes the proof of \eqref{megjelolt}.

By \ax{IND}-Main--Value theorem there is a $z\in (x_\beta,y_\beta)$ such that
\begin{equation}
(i\circ j)'(z)=\frac {(i\circ j)(x_\beta)-(i\circ j)(y_\beta)}{x_\beta-y_\beta}=\frac{x_\gamma-y_\gamma}{x_\beta-y_\beta}.
\end{equation}
By this and \eqref{megjelolt}, we conclude that $\big|\frac{x_\gamma-y_\gamma}{x_\beta-y_\beta}\big|>1$.
Hence $|x_\gamma-y_\gamma|>|x_\beta-y_\beta|$, as desired.
\end{proof} 

\theoremstyle{definition} \newtheorem*{bolzano}{\colorbox{thmbgcolor}{\textcolor{thmcolor}{\ax{IND}-Bolzano's Theorem}}} 
\begin{bolzano}
Assume \ax{IND}.
Let $f:\Q\parrow\Q$ be definable and continuous such that $\dom f$ is connected, and let $a,b\in\dom f$.
If $c\in\big(f(a),f(b)\big)$, there is an $s\in(a,b)$ such that $f(s)=c$.\qed 
\end{bolzano}

\theoremstyle{definition} \newtheorem*{darboux}{\colorbox{thmbgcolor}{\textcolor{thmcolor}{\ax{IND}-Darboux's Theorem}}} 
\begin{darboux}
Assume \ax{IND}.
Let $f:\Q\parrow\Q$ be definable and differentiable such that $\dom f$ is connected, and let $a,b\in\dom f$.
If $c\in\big(f'(a),f'(b)\big)$, there is an $s\in(a,b)$ such that $f'(s)=c$.\qed
\end{darboux}

\theoremstyle{definition} \newtheorem*{meanvalue}{\colorbox{thmbgcolor}{\textcolor{thmcolor}{\ax{IND}-Mean--Value Theorem}}} 
\begin{meanvalue}
Assume \ax{IND}.
Let $f:\Q\parrow\Q$ be definable and differentiable such that $\dom f$ is connected, and let $a,b\in\dom f$.
If $a\neq b$, there is an $s\in(a,b)$ such that $f'(s)=\frac{f(b)-f(a)}{b-a}$.\qed
\end{meanvalue}

\theoremstyle{definition} \newtheorem*{role}{\colorbox{thmbgcolor}{\textcolor{thmcolor}{\ax{IND}-Role's Theorem}}} 
\begin{role}
Assume \ax{IND}.
Let $f:\Q\parrow\Q$ be definable and differentiable such that $\dom f$ is connected, and let $a,b\in\dom f$.
If $a\neq b$ and $f(a)=f(b)$, there is an $s\in(a,b)$ such that $f'(s)=0$.\qed
\end{role}

\begin{rem}
We note that \ax{IND} is not strong enough to prove every theorem of real analysis, for example, the statement that there is a function $f$ such that $f'(x)=f(x)$.
\end{rem}

\begin{rem}
Lemma \ref{lem-main} remains true even if we substitute ``$=$'' or ``$\ge$'' for ``$>$''.
The proof can be achieved by the same substitution in the original proof.
\end{rem}

\section{Concluding remarks}

We have proved several qualitative versions of GTD from a weak axiom system of SR (\ax{AccRel}) by the use of EEP.
It is important to note that the axioms of \ax{AccRel} and EEP have different statuses herein.
EEP is not an axiom, it is just a guiding principle.

The theorems of this paper can be interpreted as saying that
observers will experience time dilation in the direction of gravitation
by the corresponding measuring methods (photon, radar, Minkowski)
if all the axioms of \ax{AccRel} are true in ``our world'' and EEP is a ``good'' principle.

Since gravitation can be defined by the acceleration of dropped \emph{inertial} bodies, EEP can be formulated within \ax{AccRel}.
It raises the possibility of checking within \ax{AccRel} how good a principle EEP is.
We may be able to prove EEP from \ax{AccRel} for all observers.
On the other hand if the formulated EEP is not a theorem of \ax{AccRel}, we can ask what other axioms need to be added to \ax{AccRel} to prove EEP.

\section*{ACKNOWLEDGEMENTS}

Thanks go to Andr\'eka Hajnal for our many valuable discussion on the subject and her useful comments leading to the present version of this paper.

This work was supported by the
Hungarian National Foundation for Scientific Research Grant
T43242 and by a Bolyai Grant for Judit X.\ Madar\'asz.

\bibliographystyle{plain}

\end{document}